\documentclass[10pt]{elsarticle}
\usepackage{amssymb}
\usepackage{amsmath,amssymb,amstext}
\usepackage{amsfonts}
\usepackage{soul}
\usepackage[figuresright]{rotating}
\usepackage{lineno,hyperref}
\usepackage[dvipsnames]{xcolor}
\usepackage{tikz}
\usepackage{algorithm,algorithmicx}
\usepackage{graphicx}
\usepackage{subcaption}
\usepackage{algpseudocode}
\usepackage{pifont}
\usepackage{mathtools}
\usepackage{multirow}
\usepackage{float }
\usepackage{array}
\usepackage{eqparbox}
\usepackage{fancyvrb}

\usepackage{lscape}
\usepackage{newlfont}
\usepackage{amsthm}

\newtheorem{theorem}{Theorem}[section]
\newtheorem{lemma}{Lemma}[section]

\newtheorem{definition}{Definition}[section]

\newtheorem{remark}{Remark}[section]

\newcommand{\etal }{\mbox{\emph{et al. }}}
 
\journal{Communications in Nonlinear Science and Numerical Simulation}

\begin{document}

\begin{frontmatter}
\title{A fractional-order compartmental model for the Spread of the {COVID-19}  pandemic}
\author{T. A. Biala*}
\ead{*tab7v@mtmail.mtsu.edu}
\author{A. Q. M. Khaliq}
\ead{Abdul.Khaliq@mtsu.edu}
\address{ Center for Computational Science and Department of Mathematical Sciences,\\
Middle Tennessee State University, \\
 Murfreesboro, TN 37132-0001, USA.}

\begin{abstract}
\quad 
We propose a time-fractional compartmental model (SEI$_A$I$_S$HRD) comprising of the susceptible, exposed, infected (asymptomatic and symptomatic), hospitalized, recovered and dead population for the {{COVID-19}} pandemic.  We study the properties and dynamics of the proposed model. The conditions under which the disease-free and endemic equilibrium points are asymptotically stable are discussed. Furthermore, we study the sensitivity of the parameters and use the data from Tennessee state (as a case study) to discuss identifiability of the parameters of the model. The non-negative parameters in the model are obtained by solving inverse problems with empirical data from California, Florida, Georgia, Maryland, Tennessee, Texas, Washington and Wisconsin. The basic reproduction number is seen to be slightly above the critical value of one suggesting that stricter measures such as the use of face-masks, social distancing, contact tracing, and even longer stay-at-home orders need to be enforced in order to mitigate the spread of the virus. As stay-at-home orders are rescinded in some of these states, we see that the number of cases began to increase almost immediately and may continue to rise until the end of the year 2020 unless stricter measures are taken.  
\end{abstract}
\end{frontmatter}
\section{Introduction}
Fractional differential equations (FDEs) are used to model complex phenomena such as such as the modeling of memory-dependent phenomena ({Di Giuseppe \etal } \cite{DiGiuseppe2010}, Baleanu \emph{et al.} \cite{Baleanu2016}, Podlubny \cite{Podlubny1999}), mechanical properties of materials (Caputo and Mainardi \cite{Caputo1971}), anomalous diffusion in porous media (Fomin \etal \cite{Fomin2011}, Metzler  and Klafter\cite{Metzler2000}), groundwater flow problems (Cloot and Botha \cite{Cloot2006}, Iaffaldano \etal \cite{Iaffaldano2005}),  and control theory (Podlubny \cite{Podlubny1994}), among others. They serve as a generalization of the integer-order differential equations and give more degree of freedom for modeling of biological and physical processes. FDEs have been applied in biological tissues \cite{Magin2010}, DNA sequencing \cite{Machado2011b}, Pine Wilt disease \cite{Shah2020}, lung tissue  mechanics and models \cite{Ionescu2017} , harmonic oscillators \cite{Baleanu2020}, Dengue fever \cite{Diethelm2013}, measles \cite{Islam2014}, human liver \cite{Baleanu2020b}, diffusion processes \cite{Sierociuk2015}, SEIR models \cite{Almeida2018}.
Infectious disease outbreaks  are one of the main causes of deaths in human. Their dynamics and spread are modeled and studied before the introduction of vaccines. {COVID-19}  began in December 2019 in China. Over the last few months, it has spread rapidly leading to over 400,000 deaths across the globe. The first occurrence in the United States was seen around mid  January in Washington \cite{NYT} and has spread across America with over 100,000 deaths and 1.5 million infected. The pandemic has disrupted the day-to-day activities of the human life with over six million jobs lost in the United States. Several actions and measures have been taken by the federal, state and local governments to mitigate the spread of the pandemic. The most prominent measures taken include social distancing, testing, use of face-masks and contact tracing. It is important to model this pandemic in order to better understand the spread and dynamics as well as address the challenges of the pandemic. In short, mathematical models are important to guide the decisions of health  and government officials. \\
The goal of this study is to examine and analyze the spread of the pandemic using a modification of the Susceptible-Exposed-Infected-Recovered (SEIR) model with a time-fractional derivative.  The use of fractional derivatives in the model stems from the fact that the spread of infectious diseases depends not only on the current state but also on its past states (history or memory dependency). Additionally, time-fractional order models reduce errors resulting from neglect of parameters in models. We shall focus on some selected states in the US. We note that models that consider the US as a whole may be misleading and have limited applicability as different states have different {economical and political perspectives which determines the different control strategies used for each state.}  For example, while some states such as Maryland,  New Jersey, New York, Connecticut, among others,  enforced the use of masks in public places and longer stay-at-home order \cite{CNN}, other states  do not enforce these measures thereby allowing for a possibility of increase of infected individuals in such states.  There have been several models for the study of the pandemic. Lu \etal \cite{Lu2020} considered a fractional-order SEIHDR model which incorporates intercity movements. Liu \etal \cite{Liu2020} studied the dynamics of the pandemic by considering asymptomatic and symptomatic infected populations separately. {Giordano \etal \cite{Giordano2020} studied the COVID-19 epidemic with intervention strategies in Italy. They proposed a model consisting of different stages: susceptible, infected, diagnosed, ailing, recognized, threatened, healed and extinct. They further discuss the long time behavior of the populations in which the susceptible, healed and extinct population remains. Stella \etal \cite{Stella2020} studied the role of asymptomatic individuals via complex networks. In particular, they formulated a model which aims at studying the interactions in the population by means of complex networks. They further extended the model to a structured nonhomogenous version by means of the Watts-Strogatz complex network.} Wu \etal \cite{Wu2020} studied domestic and international spread of the pandemic by using different data sets. Zhao and Chen \cite{Zhao2020} discussed the dynamics of the pandemic by considering the Susceptible, unquarantined infected, quarantined infected and Confirmed infected (SUQC) model and parametrize the intervention effect of control measures. Zhang \etal \cite{Zhang2020} considered a fractional SEIR model with different order of the time-fractional derivative for each of the different population being studied. {Tuan \etal \cite{Tuan2020} proposed a fractional-order model using the Caputo derivative for studying the transmission of COVID-19. They discussed the existence and uniqueness of solutions to the proposed model. They further used the generalized Adams-Bashforth-Moulton method for simulating the model.  Bahloul \etal \cite{Bahloul2020} proposed a fractional-order Susceptible-Exposed-Infected-Quarantined- Recovered-Death-Insusceptible (SEIQRDP) model for predicting the spread of COVID-19.} \\
The remainder of the paper is organized as follows: In section 2, we give some preliminary results and definitions from fractional Calculus. Furthermore, we discuss the formulation of the model by considering  seven compartments: Susceptible-Exposed-Asymptomatic infected-Symptomatic infected, Hospitalized, Recovered and Dead populations (SEI$_A$I$_S$HRD). Section 3 details the properties and theoretical analysis of the model. Section 4 discusses the parameter sensitivity and identifiability analysis. In section 5, we applied the model to observed data for some selected states in the US. In particular, this section details solving several inverse problems for parameter estimation and computation of the basic reproduction number. Finally, we give some concluding remarks in Section 6.
\section{Model Setup}
\subsection{Preliminaries}
In this section, we present some preliminary results and definitions in fractional calculus.
\begin{definition}
\cite{Podlubny1999}
The gamma function $\Gamma(\alpha)$ is defined by the integral
\[ \Gamma(\alpha) = \int_{0}^{\infty} e^{-x} x^{\alpha}\,dx\]
which converges in the right half of the complex plane $\mathcal{R}e(z) > 0$.
\end{definition}
\begin{definition}
\cite{Caputo1967}
For any $t>0$, the Caputo-fractional derivative of order $\alpha,~ (n < \alpha \leq n-1)$ of a function $f(t)$ is defined as 
\[ _{t_0}\mathcal{D}_t^\alpha f(t) = \dfrac{1}{\Gamma(n-\alpha)}\int_{t_0}^{t}(t-\tau)^{n-\alpha-1}f^{(n)}(\tau)\,d\tau.\]
\end{definition}
\begin{definition}\cite{Podlubny1999}
The Mittag-Leffler  function which generalizes the exponential function for fractional calculus is defined as
\[ E_{\alpha, \beta}(z) = \sum_{k=0}^{\infty}\dfrac{z^k}{\Gamma(\alpha k + 1)}, ~~~\alpha \in \mathbb{R}^+,~z \in \mathbb{C}.\]
\end{definition}
\begin{remark}
\item More generally, the two parameter Mittag-Leffler function is defined as 
\[ E_{\alpha}(z) = \sum_{k=0}^{\infty}\dfrac{z^k}{\Gamma(\alpha k + \beta)}, ~~~\alpha, \beta \in \mathbb{R}^+,~z \in \mathbb{C}.\]
It has the following properties:
\begin{enumerate}
\item $E_{\alpha, \beta}(z) = zE_{\alpha, \alpha+\beta}(z) + \dfrac{1}{\Gamma(\beta)}$.
\item $_0\mathcal{D}_t^\alpha e^{\lambda t} = t^{-\alpha}E_{1, 1-\alpha}(\lambda t).$
\item $_0\mathcal{D}_t^\alpha E_{\alpha, 1}(\lambda t^\alpha) = \lambda E_{\alpha, 1}(\lambda t^\alpha)$.
\end{enumerate}
\end{remark}
\begin{definition}\cite{Li2009}
A point $x^*$ is said to be an equilibrium point of the system $_{t_0}\mathcal{D}_t^\alpha = f(t, x(t)), ~x(t_0) > 0$ if and only if $f(t, x^*(t)) = 0$.
\end{definition}
\begin{definition}
\cite{Lu2020} An equilibrium point $x^*$ of the system $_{t_0}\mathcal{D}_t^\alpha x(t) = f(t, x(t)),~ x(t_0) > 0$  is said to be asymptotically stable if all {the} eigenvalues of the Jacobian matrix $J = \partial f/\partial x$, evaluated at the equilibrium point, satisfies $|arg(\lambda_i)| > \dfrac{\alpha \pi}{2}$, where $\lambda_i$ are the eigenvalues of $J$.
\end{definition}
\subsection{Model Formulation}
The model discussed in this work is a modification of the SEIR model having three additional compartments. We consider a SEI$_A$I$_S$HRD compartmental model which comprises of the susceptible, exposed, infected (asymptomatic and symptomatic), hospitalized, recovered and dead population. We assume that the natural death and birth rates are the same. We further assume that deaths in the S and R compartments are due to natural deaths and deaths in the  other compartments are as a result of the pandemic. The schema below shows the transmission flow of the model.
\begin{figure}[H]
    \centering
    \includegraphics[width=1.20\textwidth]{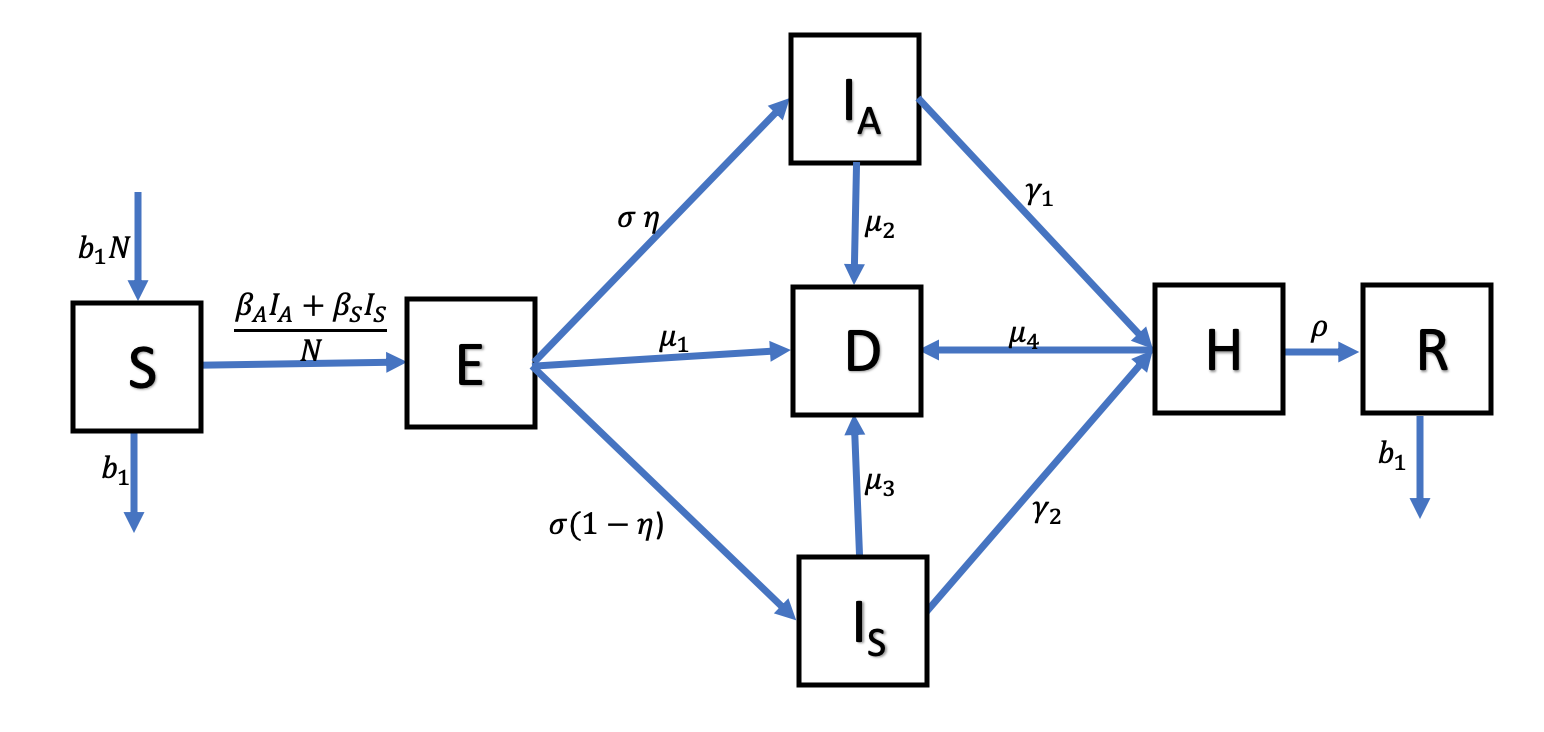}
    \caption{Schematic diagram of the proposed SEI$_A$I$_S$HRD model}
    \label{fig:mschema}
\end{figure}
Thus, the model consists of the following system of ODEs
\begin{equation*}
    \begin{split}
        _0\mathcal{D}^\alpha_{t}S(t) &= b_1 N - \dfrac{S(t)}{N}\left(\beta_A I_A(t) + \beta_S I_{S}(t)\right) - b_1 S(t),\\
_0\mathcal{D}^\alpha_{t}E(t) &= \dfrac{S(t)}{N}\left(\beta_A I_A(t) + \beta_S I_{S}(t)\right) - (\sigma + \mu_1) E(t),\\
_0\mathcal{D}^\alpha_{t}I_A(t) &= \eta\sigma E(t) - (\gamma_1 + \mu_2)I_A(t), \\
_0\mathcal{D}^\alpha_{t}I_S(t) &= (1-\eta)\sigma E(t) - (\gamma_2 + \mu_3) I_S(t),\\
_0\mathcal{D}^\alpha_{t}H(t) &= \gamma_1 I_A(t) + \gamma_2 I_S(t)  - (\rho + \mu_4) H(t),\\
_0\mathcal{D}^\alpha_{t}R(t) &= \rho H(t) - b_1 R(t),\\
_0\mathcal{D}^\alpha_{t}D(t) &= \mu_1 E(t) +  \mu_2 I_A(t)  + \mu_3 I_S(t) +  \mu_4 H(t).
    \end{split}
\end{equation*}
The system has the associated initial data
\begin{eqnarray*}
S(0) = S_0 \geq 0, & E(0) = E_0 \geq 0, & I_A(0) = I_{A0} \geq 0, \\
I_S(0) = I_{S0} \geq 0, & H(0)  = H_0 \geq 0, & R(0) = R_0 \geq 0,\\
D(0) = D_0 \geq 0.
\end{eqnarray*}
The total population is N which is further divided into $S(t), ~E(t)~, I_A(t),~I_S(t)$, $H(t)$ and $R(t)$, $b_1$ is the natural birth rate, $\beta_A$ and $\beta_S$ are the transmission rates to the susceptible population from the asymptomatic and symptomatic populations, respectively. $1/\sigma$ is the incubation period of an exposed individual, $\eta$ denotes the fraction of the exposed population that becomes asymptomatic after the incubation period and the remaining $(1-\eta)$ of the population are symptomatic.  $\gamma_1$ and $\gamma_2$ are the infectious rates of an asymptomatic and a symptomatic individual, respectively. $\rho$ is the recovery rate through hospitalization and,  $\mu_1,~\mu_2,~ \mu_3$ and $\mu_4$ are the mortality rates of the exposed, asymptomatic, symptomatic and hospitalized populations, respectively. We note that the death population in the model comprises of deaths during exposure $\mu_1E(t)$, infectious period ($\mu_2I_A$ and $\mu_3I_S$) and hospitalization $\mu_4H$, and assume that deaths due to other natural occurrences are negligible for this populations.\\
We note that the parameters of the model are non-negative and  have dimensions given by  1/time$^\alpha$. This observation was originally noted in Diethelm \cite{Diethelm2013}. To alleviate this difference in dimensions, we replace the parameters (except $\eta$) with a power $\alpha$ of new parameters to obtain the new system of equations
\begin{equation}
\label{eqn:model}
    \begin{split}
        _0\mathcal{D}^\alpha_{t}S(t) &= b_1^\alpha N - \dfrac{S(t)}{N}\left(\beta_A^\alpha I_A(t) + \beta_S^\alpha I_{S}(t)\right) - b_1^\alpha S(t), \\
_0\mathcal{D}^\alpha_{t}E(t) &= \dfrac{S(t)}{N}\left(\beta_A^\alpha I_A(t) + \beta_S^\alpha I_{S}(t)\right) - (\sigma^\alpha + \mu_1^\alpha) E(t),\\
_0\mathcal{D}^\alpha_{t}I_A(t) &= \eta\sigma^\alpha E(t) - (\gamma_1^\alpha + \mu_2^\alpha)I_A(t), \\
_0\mathcal{D}^\alpha_{t}I_S(t) &= (1-\eta)\sigma^\alpha E(t) - (\gamma_2^\alpha + \mu_3^\alpha) I_S(t),\\
_0\mathcal{D}^\alpha_{t}H(t) &= \gamma_1^\alpha I_A(t) + \gamma_2^\alpha I_S(t)  - (\rho^\alpha + \mu_4^\alpha) H(t),\\
_0\mathcal{D}^\alpha_{t}R(t) &= \rho^\alpha H(t) - b_1^\alpha R(t),\\
_0\mathcal{D}^\alpha_{t}D(t) &= \mu_1^\alpha E(t) +  \mu_2^\alpha I_A(t)  + \mu_3^\alpha I_S(t) +  \mu_4^\alpha H(t).
    \end{split}
\end{equation}
\section{Model Analysis}
In this section, we discuss the properties of the model beginning with the existence, uniqueness, non-negativity and boundedness of solutions of the model (\ref{eqn:model}). For simplicity in analysis, we reduce the system (\ref{eqn:model}) to
\begin{equation}
\label{eqn:model2}
    \begin{split}
        _0\mathcal{D}^\alpha_{t}S(t) &= b_1^\alpha N - \dfrac{S(t)}{N}\left(\beta_A^\alpha I_A(t) + \beta_S^\alpha I_{S}(t)\right) - b_1^\alpha S(t), \\
_0\mathcal{D}^\alpha_{t}E(t) &= \dfrac{S(t)}{N}\left(\beta_A^\alpha I_A(t) + \beta_S^\alpha I_{S}(t)\right) - (\sigma^\alpha + \mu_1^\alpha) E(t),\\
_0\mathcal{D}^\alpha_{t}I_A(t) &= \eta\sigma^\alpha E(t) - (\gamma_1^\alpha + \mu_2^\alpha)I_A(t), \\
_0\mathcal{D}^\alpha_{t}I_S(t) &= (1-\eta)\sigma^\alpha E(t) - (\gamma_2^\alpha + \mu_3^\alpha) I_S(t),\\
_0\mathcal{D}^\alpha_{t}H(t) &= \gamma_1^\alpha I_A(t) + \gamma_2^\alpha I_S(t)  - (\rho^\alpha + \mu_4^\alpha) H(t),\\
_0\mathcal{D}^\alpha_{t}R(t) &= \rho^\alpha H(t) - b_1^\alpha R(t),
    \end{split}
\end{equation}
since $D$ is a linear combination of populations in some of the other compartments.
\begin{theorem}
There exist a unique solution to the system (\ref{eqn:model2}) and the solution is non-negative and bounded for any given initial data $(S_0, E_0, I_{A0}, I_{S0}, H_0, R_0) \geq \mathbf{0} \in \mathbb{R}^6$.
\end{theorem}
\begin{proof}
See Appendix A.
\end{proof}
\subsection{Stability Analysis}
\subsubsection{Computation of the basic reproduction number $R_0$}
We shall use the next generation matrix  originally proposed by Diekmann \etal \cite{Diekmann1990} and further elaborated on by van den Driesche and Watmough \cite{VanDriessche2002} and Diekmann \etal \cite{Diekmann2009} to determine $R_0$. {The disease-free equilibrium point (DFE) is $\left(N, 0, 0, 0, 0,0 \right)$}. Now, consider the four compartments $Y = (Y_1, Y_2, Y_3, Y_4) = (E, I_A, I_S, H)$ containing the infected individuals and let $Y^*$ be the DFE point. Since the DFE point exists and is stable (shown in the next section) in the absence of any disease, then the linearized equation at the DFE is
\[ _0\mathcal{D}^\alpha_{t} Y_i = \mathcal{F}_i(Y) - \mathcal{V}_i(Y),~ i=1(1)4, \]
where $\mathcal{F}_i(Y)$ is the rate of appearance of new infections in compartment $i$ and $\mathcal{V}_i(Y)$ is the rate of transfer of infections to and from compartment $i$. We further define 
\[F = \dfrac{\partial \mathcal{F}(Y)}{\partial Y_j}\big|_{Y=Y^*}~~\text{and}~~ V = \dfrac{\partial \mathcal{V}_i(Y)}{\partial Y_j}\big|_{Y=Y^*},~~~i,~j = 1(1)4. \]
Then $\rho(FV^{-1})$ is the basic reproduction number $R_0$, where $\rho(x)$ is the spectral radius of $x$ and $FV^{-1}$ is the next generation matrix. For the system (\ref{eqn:model2}),
\begin{equation*}
F = \left[
    \begin{array}{cccc}
     0 &  \beta_A^\alpha &  \beta_S^\alpha & 0\\
     0 & 0 & 0 & 0 \\
     0 & 0 & 0 & 0 \\
     0 & 0 & 0 & 0 
\end{array}
\right], ~~~
V = \left[
    \begin{array}{cccc}
     \sigma^\alpha + \mu_1^\alpha  & 0 & 0 &0  \\
     -\eta \sigma^\alpha & \gamma_1^\alpha + \mu_2^\alpha & 0 & 0\\
     - (1- \eta)\sigma^\alpha & 0 & \gamma_2^\alpha + \mu_3^\alpha & 0\\
     0 & -\gamma_1^\alpha & - \gamma_2^\alpha & (\rho^\alpha + \mu_4^\alpha)
\end{array}
\right]
\end{equation*}
and the basic reproduction number is given as
\begin{equation}
\label{eqn:Ro}
R_0 = \dfrac{\sigma^\alpha \left[\eta \beta_A^\alpha(\gamma_2^\alpha + \mu_3^\alpha) + (1 - \eta)\beta_S^\alpha(\gamma_1^\alpha + \mu_2^\alpha) \right]}{(\sigma^\alpha + \mu_1^\alpha)(\gamma_1^\alpha + \mu_2^\alpha)(\gamma_2^\alpha + \mu_3^\alpha)}.
\end{equation}
\begin{lemma}
The fractional system (\ref{eqn:model2}) has at most two equilibrium points
\begin{enumerate}
    \item a disease-free equilibrium point  DFE =  $\left(N, 0, 0, 0, 0,0 \right)$
    \item an endemic equilibrium point EE = $(S^*, E^*, I_A^*, I_S^*, H^*, R^*)$,
\end{enumerate}
where 
\begin{equation*}
\begin{split}
S^* &= \dfrac{S_0}{R_0}, \\
E^* &= \dfrac{b_1^\alpha S_0}{R_0(\sigma^\alpha + \mu_1^\alpha)}\left(1 - \dfrac{1}{R_0} \right),\\
I_A^* &= \dfrac{\eta b_1^\alpha\sigma^\alpha S_0}{R_0(\sigma^\alpha + \mu_1^\alpha)(\gamma_1^\alpha + \mu_2^\alpha)}\left(1 - \dfrac{1}{R_0} \right),\\
I_S^* &= \dfrac{(1- \eta)b_1^\alpha\sigma^\alpha S_0}{R_0(\sigma^\alpha + \mu_1^\alpha)(\gamma_2^\alpha + \mu_3^\alpha)}\left(1 - \dfrac{1}{R_0} \right),\\
 H^* &= \dfrac{b_1^\alpha\sigma^\alpha S_0}{R_0(\sigma^\alpha + \mu_1^\alpha)(\rho^\alpha + \mu_4^\alpha)}\left(1 - \dfrac{1}{R_0} \right)\left(\dfrac{\eta \gamma_1^\alpha}{\gamma_1^\alpha} + \dfrac{(1-\eta)\gamma_2^\alpha}{\gamma_2^\alpha + \mu_3^\alpha} \right),\\
 R^* &=  \dfrac{\rho^\alpha b_1^\alpha\sigma^\alpha S_0}{R_0(\sigma^\alpha + \mu_1^\alpha)(\rho^\alpha + \mu_4^\alpha)}\left(1 - \dfrac{1}{R_0} \right)\left(\dfrac{\eta \gamma_1^\alpha}{\gamma_1^\alpha} + \dfrac{(1-\eta)\gamma_2^\alpha}{\gamma_2^\alpha + \mu_3^\alpha} \right).
 \end{split}
\end{equation*}
\end{lemma}
\begin{theorem}
The DFE point is locally asymptotically stable if $R_0 < 1$ and $K < 1$, where 
\[K = \dfrac{\sigma^\alpha \left[\eta \beta_A^\alpha(\sigma^\alpha + \gamma_1^\alpha + \mu_1^\alpha + \mu_2^\alpha) + (1 - \eta)\beta_S^\alpha(\sigma^\alpha + \gamma_2^\alpha + \mu_1^\alpha + \mu_3^\alpha) \right]}{\left(\sigma^\alpha + \gamma_1^\alpha + \mu_1^\alpha + \mu_2^\alpha\right)\left(\sigma^\alpha + \gamma_2^\alpha + \mu_1^\alpha + \mu_3^\alpha\right)\left(\gamma_1^\alpha + \gamma_2^\alpha + \mu_2^\alpha + \mu_2^\alpha\right)}.\]
\end{theorem}
\begin{proof}
See Appendix B.
\end{proof}

\begin{theorem}
The EE point is locally asymptotically stable if $R_0 > 1$ and $(A_1A_2 -A_3)A_3 \geq A_1^2A_4$, where
\begin{equation*}
    \begin{split}
    A_1 &= b_1^\alpha(2R_0 - 1) + R_0(\sigma^\alpha + \gamma_1^\alpha + \gamma_2^\alpha + \mu_1^\alpha + \mu_2^\alpha + \mu_3^\alpha),\\
        A_2 &=  -\sigma^\alpha R_0(\eta \beta_A^\alpha + (1-\eta)\beta_S^\alpha) + b_1^\alpha R_0(2R_0 - 1)(\sigma^\alpha + \gamma_1^\alpha + \gamma_2^\alpha + \mu_1^\alpha + \mu_2^\alpha + \mu_3^\alpha)\\
        & ~~+ R_0^2[(\gamma_2^\alpha + \mu_3^\alpha)(\sigma^\alpha + \mu_1^\alpha) + (\gamma_2^\alpha + \mu_3)(\gamma_1^\alpha + \mu_2^\alpha) + (\sigma^\alpha + \mu_1^\alpha)(\gamma_1^\alpha + \mu_2^\alpha)],\\
        A_3 &= b_1^{2\alpha}R_0^2(2R_0 - 1)[(\gamma_2^\alpha + \mu_3^\alpha)(\sigma^\alpha + \mu_1^\alpha) + (\gamma_2^\alpha + \mu_3)(\gamma_1^\alpha + \mu_2^\alpha) + (\sigma^\alpha + \mu_1^\alpha)(\gamma_1^\alpha + \mu_2^\alpha)]\\ 
        & ~~+ R_0(\sigma^\alpha + \mu_1^\alpha)(\gamma_1^\alpha + \mu_2^\alpha)(\gamma_2^\alpha + \mu_3) - \sigma^\alpha((1-\eta)(\gamma_1^\alpha + \mu_2^\alpha)\beta_S^\alpha + \eta(\gamma_2^\alpha + \mu_3^\alpha)\beta_A^\alpha)),\\
A_4 &= b_1^{3\alpha}R_0^3(2R_0 -1)(\sigma^\alpha + \mu_1^\alpha)(\gamma_1^\alpha + \mu_2^\alpha)(\gamma_2^\alpha + \mu_3^\alpha) - \sigma^\alpha((1-\eta)\beta_S^\alpha(\gamma_1^\alpha + \mu_2^\alpha)\\
& ~~+ \eta\beta_A^\alpha(\gamma_2^\alpha + \mu_3^\alpha))]
    \end{split}
\end{equation*}
\end{theorem}
\begin{proof}
See Appendix C.
\end{proof}
\section{Parameter Sensitivity and Identifiability Analysis}
We discuss the sensitivity and identifiability of the parameters with respect to the proposed model. 
\subsection{Sensitivity analysis}
The sensitivity analysis (SA) deals with the significance or importance of the parameters in the model. In particular, it finds the most influential parameters that drives the dynamics of the model. It also describes the extent to which  parameter changes affects the result of the methods or models with the goal of identifying the best set of parameters that describes the process or phenomena in question.\\
There are several SA methods which are broadly classified as local and global methods. In this work, we shall focus on the Morris screening method (local method) and Sobol analysis method (global method).
\subsubsection{Morris Screening Method}
The Morris screening method is a local sensitivity measure that makes use of the first order derivative of an output function $y = f(\theta) = f(\theta_1, \cdots, \theta_p)$ with respect to the input parameter $\theta$. It measures the effect of the output when the input variable is perturbed one at a time around a nominal value. It serves as a first check, in most analysis, in screening parameters for identifiability. The method evaluates elementary effects \cite{Saltelli2000, Saltelli2004, Ye2017} with the $i$th parameter through the forward perturbation
\[ g_i(\theta) = \dfrac{f(\theta_1, \theta_2, \cdots, \theta_i + \Delta \theta_i, \cdots, \theta_p) - f(\theta_i, \cdots, \theta_p)}{\Delta \theta_i}, ~~~ i=1(1)p\]
Morris \cite{Morris1991} proposed two sensitivity measures, the mean ($\mu$) and the standard deviation ($\tilde{\sigma}$) of the elementary effects. For non-monotonic models, $\mu$ may lead to a very small value due to cancellation effects. For this reason, Campolongo \etal \cite{Campolongo2007} proposed the use of absolute values for evaluating the mean. In order to obtain a dimension-free sensitivity, we prefer the use of the sensitivity measure $\delta$ given in Brun \etal \cite{Brun2001} as
\[\delta_i =  \sqrt{\dfrac{1}{N}\sum^{N}_{j=1}\tilde{g}_{_{ij}}^2}, ~~~~~~~i = 1(1)p, ~\text{and}~ j=1(1)N,\]
where N is the number of sample points and 
\[ \tilde{g}_i(\theta) = \dfrac{f(\theta_1, \theta_2, \cdots, \theta_i + \Delta \theta_i, \cdots, \theta_p) - f(\theta_1, \cdots, \theta_p)}{\Delta \theta_i}\dfrac{\theta_i}{f(\theta_1,\cdots, \theta_p)}. \]
A common practice in the literature \cite{Pohlmann2007, Ye2016, Ye2017} is to plot the indices $\delta$ against $\tilde{\sigma}$, the standard deviation. We see from Fig.\,2(a) that the fractional order $\alpha$ has the highest influence on the model output over time. The transmission rates for the symptomatic and asymptomatic populations and, the fractional parameter $\eta$ are the next most influential parameters in the model. This is further corroborated by Fig.\,2(b). The four parameters ($\alpha$ $\beta_S, \beta_A$ and $\eta$) denoted by red squares are the most important parameters, that is they have the largest $\delta$ and $\tilde{\sigma}$ values. The parameters ($b_1, \rho, \mu_1, \mu_2, \mu_3$ and $\mu_4$) represented by the blue squares have the least influence on the model output and can be considered unimportant. The other parameters represented by the green squares have more influence than the parameters represented by the blue squares.\\
One major setback of the Morris screening test for sensitivity analysis is the consideration of each parameter individually and independently of the other parameters. In real applications, this is not true as parameters have collinearity and dependencies on one another.
\begin{figure}[H]
    \centering
    \begin{subfigure}[b]{.8\textwidth}
    \centering
    \includegraphics[width=1.2\textwidth]{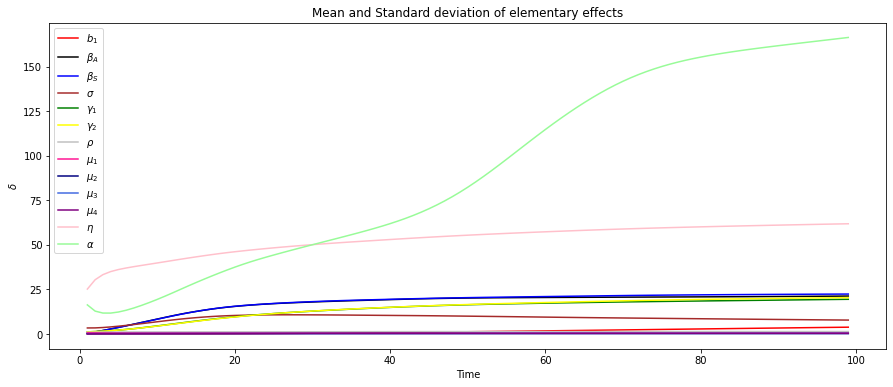} 
    \subcaption{Sensitivity of the parameters over time} 
  \end{subfigure}
  \\
   \begin{subfigure}[b]{.8\textwidth}
    \centering
    \includegraphics[width=1.2\textwidth]{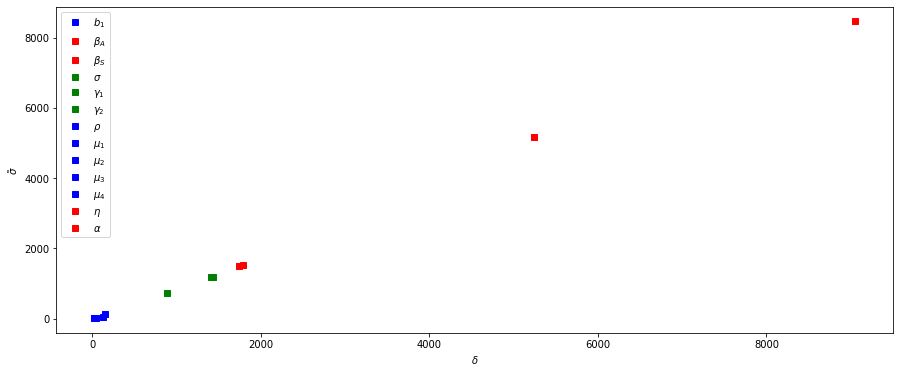}
    \subcaption{Parameter Importance}
  \end{subfigure}
  \caption{Morris screening test }
\end{figure}
\subsubsection{Sobol Analysis}
The Sobol method is a variance-based sensitivity analysis method which unlike the Morris screening method takes into account the effect of the relationship between each parameters of the model. It uses the decomposition of variance to calculate Sobol's sensitivity indices: first and total order sensitivity measures. The basic idea of the Sobol's method is the decomposition of the model output  function $y = f(\theta_1, \cdots, \theta_p)$ into summands of increasing dimensionality, that is
\[ V(y) = V_{1, \cdots, p} + \sum_{i=1}^{p}V_i + \sum^{p}_{i=1}\sum_{j>1}^{p}V_{i, j} + \cdots\]
where $V_i$ is the partial variance of the contribution of the parameter $\theta_i$ and $V_{i, \cdots, s}$ is the partial variances caused by the interaction of the parameters $(\theta_1, \cdots, \theta_s)$ for $s \leq p$.\\
The first order sensitivity index measures the main effect of parameter $\theta_i$ on the model output; that is the partial contribution of $\theta_i$ to the variance $V(y)$.The index \cite{Sobol1993, Sobol2001} is defined as
\[ S_i = \dfrac{V_i}{V(y)}.\]
The larger this index, the more sensitive the parameter is to the model output \cite{Sobol1993, Sobol2001}. Using the law of total variances \cite{Sobol2001, Ye2017}, the index can also be expressed as 
\[ V(y) = V_{\theta_i}(E_{\theta_{\sim i}}(y|\theta_i)) + E_{\theta_i}(V_{\theta_{\sim i}}(y|\theta_i))\]
and 
\[ S_i = \dfrac{ V_{\theta_i}(E_{\theta_{\sim i}}(y|\theta_i))}{V(y)}\]
where $ V_{\theta_i}(E_{\theta_{\sim i}}(y|\theta_i))$ is the partial variance caused by $\theta_i$ and $E_{\theta_{\sim i}}(y|\theta_i)$ is the mean of the model output calculated by using all the values of the other parameters $\theta_{\sim i}$ (except $\theta_i$) and $V(y)$ is the total variance. \\
The total sensitivity indices \cite{Homma1996} measures the effects of parameter $\theta_i$ and the interaction with the other parameters. It is defined as
\[ S_{T_i} = \dfrac{V_i + V_{i,j} + \cdots + V_{i, j, \cdots, p}}{V(y)}.\]
The total variance, $V(y)$, for this index is given as 
\[V(y) = V_{\theta_{\sim i}}(E_{\theta_{ i}}(y|\theta_{\sim i})) + E_{\theta_{\sim i}}(V_{\theta_{i}}(y|\theta_{\sim i})) \]
and 
\begin{equation*}
    \begin{split}
        S_{T_i}  &= \dfrac{E_{\theta_{\sim i}}(V_{\theta_{i}}(y|\theta_{\sim i}))}{V(y)}\\
        & = \dfrac{V(y) - V_{\theta_{\sim i}}(E_{\theta_{ i}}(y|\theta_{\sim i}))}{V(y)}.
    \end{split}
\end{equation*}
The mean and the variance can be evaluated using quasirandom sampling method \cite{Saltelli2010, Ye2017} and are given as 
\[ V_{\theta_{i}}(E_{\theta_{ \sim i}}(y|\theta_{i})) = \dfrac{1}{N}\sum_{j=1}^{N}f(\mathbf{B}_j)\left( f(\mathbf{A}^i_{\mathbf{B}, j}) - f(\mathbf{A}_j)\right),\]
and
\[ E_{\theta_{\sim i}}(V_{\theta_{i}}(y|\theta_{\sim i})) = \dfrac{1}{2N}\sum_{j=1}^{N}\left(f(\mathbf{A}_j) - f(\mathbf{A}^i_{\mathbf{B}, j}) \right)^2,\]
where $\mathbf{A}$ and $\mathbf{B}$ are two independent parameter sample matrices of dimensions $N \times p$. We shall use the python SALib package \cite{SALib} to compute the first and total order variance indices. Fig.\,(3) shows that the fractional order $\beta_S$ has the highest interaction with the other parameters. These results are consistent with the results in the Morris screening test as  important parameters of the test  show high interaction with the other parameters.
\begin{figure}[H]
    \centering
    \includegraphics[width=1.0\textwidth]{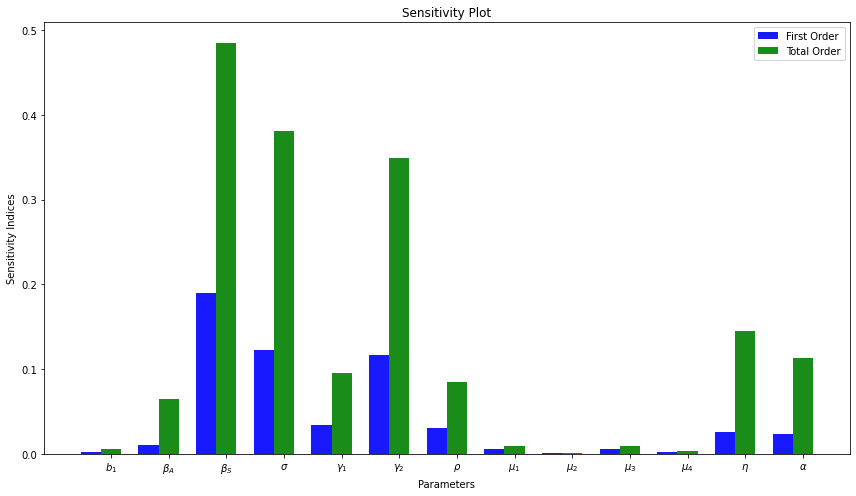}
    \caption{Sobol Sensitivity Indices}
    \label{fig:sobol}
\end{figure}
\subsection{Parameter Identifiability}
The concept of identifiability is dependent on sensitivity. It entails the selection of the subset of parameters of a model having little or no collinearity and uncertainty, and which can be identified uniquely from a given set of observed data or measurements. In other words, it answers the question "Can the available data be described by the model and the selected subset of parameters?". There are several techniques or tests for parameter identifiability. Most of the tests are based on the Fisher information matrix (FIM) $F = \chi^T\chi$  where $\chi = \partial y /\partial \theta$ for a model output function $y$. Cobelli and Di Stefano \cite{Cobelli1980} showed that a sufficient condition for identifiability is the non-singularity of FIM. Burth \etal \cite{Burth1999} proposed an iterative estimation process which implements a reduced-order estimation by finding parameters whose axis lie closest to the direction of FIM. The associated parameter values are then fixed at prior estimates during the iterated process. Brun \etal \cite{Brun2001} studied parameter identifiability using two indices; a parameter importance ranking index $\delta$ and a collinearity index $\gamma_{_K}$ which depends on the smallest eigenvalues of submatrices of $\chi^T\chi$ corresponding to the parameter subset $K$. Cintr\'on-Arias \etal \cite{Cintron-Arias2009} explained the need for a good parameter subset for identifiability to satisfy the full rank test. They further introduced two indices; the selection score and the condition number of $\chi^T\chi$. The smaller these indices the lesser the collinearity and uncertainty in the parameter values of the subset. Finally, they used the coefficient of variation index to examine the effect of parameters in the parameter subset.  In this work, we shall use the test proposed by Cintr\'on-Arias \etal in identifying the parameters. The algorithm can be summarized in the following steps

 \begin{algorithm}[H]
\caption{Algorithm for Parameter subset Selection \cite{Cintron-Arias2009}}
\label{myAlgorithm3}
\begin{algorithmic}[1]
\State Perform a combinatorial search for all possible parameter subsets. Let
\[S_p = \{\theta = (\lambda_1, \lambda_2, \cdots, \lambda_p) \in \mathbb{R}^p \big| \lambda_k \in \mathcal{I}~\text{and}~ \lambda_k \neq \lambda_m ~\forall~k, m = 1, \cdots, p\},\]
where $\mathcal{I} = \{b_1, \beta_A, \beta_S, \sigma, \gamma_1, \gamma_2, \rho, \mu_1, \mu_2, \mu_3, \mu_4, \eta, \alpha \}$.
\State Select parameter subsets that pass the full rank test; that is 
\[\Theta_p = \{\theta \big| \theta \in S_p \subset \mathbb{R}^p,~ \text{Rank}(\chi(\theta)) = p \}. \]
\State For each $\theta \in \Theta_p$, calculate the parameter selection score $\zeta(\theta) = |\vartheta(\theta)|$  where 
\[\vartheta = \dfrac{\sqrt{\Sigma (\theta)_{ii}}}{\theta_i}, ~~~i = 1, \cdots, p, \]
and $\Sigma(\theta) = \sigma_0^2 \left[ \chi^T(\theta)\chi(\theta)\right]^{-1} \in \mathbb{R}^p$.
\State Calculate the condition number $\kappa(\chi(\theta))$ for each parameter subset $\theta \in \Theta_p$.
The smaller the values of $\kappa(\chi(\theta))$
 and $\vartheta(\theta)$, the lower the uncertainty possibilities in the estimate.
\end{algorithmic}
\end{algorithm}
To discuss the results in this section, we shall use the state of Tennessee as a case study to understand parameter identifiability. Furthermore, we used the following values obtained using a random search algorithm  as the nominal parameter set $\theta_0$ for the model:
\[
    b_1 = 8.0316\text{e-07},  ~~ \beta_A = 1.0000\text{e-10}, ~~ \beta_S = 1.2312\text{e+00},~~\sigma = 9.9618\text{e-01},\]
    \[\gamma_1 = 2.6628\text{e-02}, ~~ \gamma_2 = 1.0000\text{e+00}, ~~ \rho = 6.1740\text{e-03},~~ \mu_1 = 3.3491\text{e-03},\]
 \[ \mu_2 = 1.0000\text{e-10},~~\mu_3 = 3.0859\text{e-04},~~\mu_4 = 2.1222\text{e-05},~~\eta = 1.5683\text{e-01}, \]
 \[\alpha = 1.0000\text{e+00}
\]
and the nominal error variance $\sigma_0 = 10$. We further divide the parameters into three groups according to their importance rankings discussed in the previous section:
\begin{equation*}
    \begin{split}
        S_1 &= (\beta_A, \beta_S, \eta, \alpha),\\
        S_2 &= (\sigma, \gamma_1, \gamma_2),\\
        S_3 &= (b_1, \rho, \mu_1, \mu_2, \mu_3, \mu_4),
    \end{split}
\end{equation*}
where $S_1$ and $S_3$ are the most and least influential parameter sets, respectively, while $S_2$ contains more influential parameters than $S_3$. We display some selections of the parameter subsets of size $p$ in Table \ref{tab:my_label1} where we have chosen the subsets with the smallest score values. The entries in Table \ref{tab:my_label1} are ordered with respect to the selection score $\vartheta(\theta)$ for each subset of same cardinality. A high selection score and condition number for a parameter subset indicates substantial collinearity and linear dependence, and thus is poorly identifiable even if the parameter subsets contains $S_1$, that is contains the set of most influential parameters. We observe that most of the selections in Table \ref{tab:my_label1} contains at least one element in each of the groups listed above. This shows that while parameter importance ranking is crucial in identifying parameters that drives the dynamics of a model, it does not have substantial effect in identifiability. Identifiability depends on proper selection of subsets including parameters in each of the three  groups above that describes the measurement or data.

To have an idea of the variations of the condition number and the selection score, we give a plot of these values for $p=2$ in fig.\,4 (with logarithmic scales). Good parameter combination in fig.\,4 corresponds to values in the lower left corner of the figure where the values, $\vartheta(\theta)$ and $\kappa(\chi(\theta))$, are relatively small.
 \begin{figure}{H}
     \centering
     \includegraphics[width=1.0\textwidth]{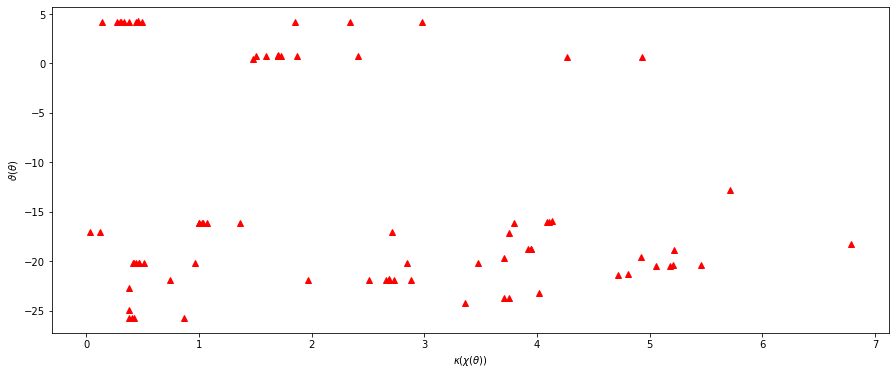}
     \caption{The condition number $\kappa(\chi(\theta))$ against the parameter selection scores $\vartheta(\theta)$ of the $N\times p$ sensitivity matrices for all parameter subsets $\theta = \Theta_p$ with $p=2$. Logarithmic scales are used on both axis }
     \label{fig:my_label4}
 \end{figure}
 \begin{table}[H]
    \centering
    \begin{tabular}{|c|c|c|c|}
    \hline
       $p$ & Parameter Subsets & $\kappa(\chi(\theta))$ 
         & $\vartheta(\theta)$\\
         \hline 
          13 &  $(b_1,\beta_A, \beta_S, \sigma, \gamma_1, \gamma_2, \rho, \mu_1, \mu_2, \mu_3, \mu_4, \eta, \alpha)$ & 9.530e+03 & 8.481e+01\\
         \hline
         12 & $(b_1,\beta_A, \beta_S, \sigma, \gamma_1, \gamma_2, \rho, \mu_1, \mu_3, \mu_4, \eta, \alpha)$ & 9.513e+03 & 6.147e+00\\
         \hline
         11 & $(b_1, \beta_S, \sigma, \gamma_1, \gamma_2, \rho, \mu_1, \mu_3,  \mu_4, \eta, \alpha)$ & 4.814e+03 & 2.479e-04\\
         \hline
          \multirow{2}{*}{10} & $(b_1, \beta_S, \sigma, \gamma_1, \gamma_2,  \rho, \mu_1, \mu_4, \eta, \alpha)$ & 3.696e+03 & 3.852e-07\\
          \cline{2-4}
           & $(b_1, \beta_S,  \sigma, \gamma_1, \gamma_2, \rho, \mu_3, \mu_4, \eta, \alpha)$ & 3.685e+03 & 4.143e-06\\
         \hline
         \multirow{2}{*}{9} & ($\beta_S, \sigma, \gamma_1, \gamma_2,  \rho, \mu_1, \mu_4, \eta, \alpha)$ & 3.411e+03 & 2.766e-07\\
          \cline{2-4}
          & $(b_1, \beta_S, \sigma, \gamma_1, \gamma_2, \rho, \mu_4, \eta, \alpha)$ & 3.311e+03 & 2.824e-07\\
          \hline
          \multirow{3}{*}{7} & $(\beta_S, \gamma_1, \gamma_2, \rho, \mu_4, \eta,  \alpha)$ & 2.718e+03 & 5.344e-08\\
          \cline{2-4}
          & $(\beta_S, \gamma_1, \gamma_2, \rho, \mu_1, \eta,  \alpha)$ & 2.947e+03 & 5.388e-08\\
          \hline
          \multirow{3}{*}{5} & $(\gamma_1,  \gamma_2, \rho, \mu_4,  \alpha)$ & 6.171e+01 & 2.006e-09\\
          \cline{2-4}
          & $(\beta_S, \gamma_1,  \rho, \mu_4, \alpha)$ & 6.389e+01 & 2.006e-09\\
          \cline{2-4}
          & $(\gamma_1,  \rho, \mu_4,  \eta, \alpha)$ & 6.989e+01 & 2.010e-09\\
          \hline
          \multirow{3}{*}{4} & $(\gamma_1, \rho, \eta,  \alpha)$ & 6.258e+01 & 3.458e-10\\
          \cline{2-4}
          & $(\gamma_1, \gamma_2, \rho, \alpha)$ & 5.289e+01 & 3.491e-10\\
          \cline{2-4}
          & $(\beta_S, \gamma_1, \rho, \alpha)$  & 5.349e+01 & 3.496e-10\\
          \hline
          \multirow{3}{*}{3} & $(\gamma_2, \rho,  \alpha)$ & 5.140e+01 & 5.431e-11\\
          \cline{2-4}
          & $(\beta_S, \rho, \alpha)$ & 5.203e+01 & 5.850e-11\\
          \cline{2-4}
          & $(\rho, \eta, \alpha)$ & 6.224e+01 & 9.109e-11\\
          \hline
          
         \end{tabular}
     \caption{Selection scores and condition numbers for some selected parameter subsets}
     \label{tab:my_label1}
 \end{table}
 To further analyze the parameter identifiability of the model, we consider the parameters subsets:
 \begin{equation*}
     \begin{split}
         \theta_1 &= (b_1, \beta_A, \beta_S, \sigma, \gamma_1, \gamma_2, \rho, \mu_1, \mu_2, \mu_3, \mu_4, \eta, \alpha),\\
         \theta_2 &= (b_1, \beta_S, \sigma, \gamma_1, \gamma_2, \rho, \mu_1, \mu_3, \mu_4, \eta, \alpha),\\
         \theta_3 &= (\beta_S, \sigma, \gamma_1, \gamma_2, \rho, \mu_1, \mu_4, \eta, \alpha),\\
         \theta_4 &=  (\beta_S, \gamma_1,  \gamma_2, \rho, \mu_4, \eta, \alpha),\\
         \theta_5 &= (\gamma_1,  \gamma_2, \rho, \mu_4, \alpha),\\
         \theta_6 &= (\gamma_1,  \gamma_2, \rho, \alpha),\\
         \theta_7 &= (\gamma_2, \rho, \alpha)
     \end{split}
 \end{equation*}
 such that $\theta_{i+1} \subset \theta_{i}, ~i = 1, \cdots, 6$. The choice of these parameter subsets are due to their relative small condition numbers and selection scores. 
In other to create synthetic data, we assume the nominal parameter subsets and error variance (given at the beginning of this section) to be the true parameter vectors and true variance. Furthermore,  we add random noise to the model output as follows:
\[Y_j = z(t_j, \theta_0) + \sigma_0 \,\mathcal{N}(0, 1), ~~j = 1, \cdots, N. \]
We solve seven inverse problems for each of the parameter subsets $\theta_i, ~i=1, \cdots, 7$. We analyze the result using the coefficient of variation and standard error \cite{Cintron-Arias2009} given as
\[ SE_j(\tilde{\theta}) == \sqrt{\tilde{\Sigma}_{j,j}}, ~~~j=1.\cdots, p\]
and
\[ v_j(\tilde{\theta})  = \dfrac{SE_j(\tilde{\theta})}{\theta_j},~~~j=1.\cdots, p,\]
where $\tilde{\Sigma}_{j,j} = \tilde{\sigma_0}^2\left[\chi(\tilde{\theta})^T \chi(\tilde{\theta})\right]^{-1}$ and $\tilde{\sigma_0}^2 = \dfrac{1}{n-p} |Y - z(\tilde{\theta})|$.

\begin{table}[]
    \centering
    \begin{tabular}{|c|c|c|}
    \hline
     $\tilde{\theta}$    & AIC & BIC  \\
    \hline
      $\theta_1$   & 164.80 & 198.67\\
              \hline
      $\theta_2$ & 163.39 & 192.05\\
        \hline
      $\theta_3$ & 150.84 & 174.28\\
              \hline
      $\theta_4$ & 144.85 & 163.09\\
              \hline
      $\theta_5$ & 142.63 & 155.68\\
              \hline
      $\theta_6$ & 140.80 & 151.22\\
              \hline
      $\theta_7$ & 141.68 & 149.50\\
              \hline
    \end{tabular}
    \caption{AIC and BIC metrics to estimate the quality of the model with different parameter sets.}
    \label{tab:my_label2}
\end{table}

From Table \ref{tab:my_label5}, it is seen that the parameters $ (\mu_2, \mu_3) \subset \theta_1$ have standard errors that is approximately $20$ times their estimates. This shows substantial uncertainty in these parameter values and any parameter subset containing these parameters may result in illogical parameter estimation from observations. The standard errors of $\beta_S, \gamma_1, \rho, \mu_3, \mu_4, \eta, \alpha$  in $\theta_2$ show  improvements and implies lower linear dependence and collinearity for parameters in $\theta_2$ than those in $\theta_1$. Thus, a substantial improvement in uncertainty quantification is seen from $\theta_1$ to $\theta_2$. Further improvements are observed for each of the other parameter subsets as more parameters are removed. For instance, with the removal  of $\beta_S$ and $\eta$ in $\theta_4$, it seen that the standard error  for $\gamma_1$ and $\gamma_2$ dropped from 4\%  and 1\% to approximately 0.07\%  and 0.003\%, respectively, of their estimates. Other improvements in $\theta_5$ include $\alpha$  where the standard error is reduced to at least one-tenth of its standard error in $\theta_4$. We note that there is no substantial gain in the removal of $\mu_4$ from $\theta_5$ and $\gamma_1$ in $\theta_6$ as seen in Table \ref{tab:label5}. \\
Parameter identifiability might be misleading without the investigation of the residual of the model \cite{Cintron-Arias2009}. The Akaike Information Criterion (AIC) and Bayesian Information Criterion (BIC) indices make use of residuals to determine the quality of models in the presence of a given set of data. Table  \ref{tab:my_label2} shows the AIC and BIC estimates for each parameter set $\theta_i, ~i=1, \cdots, 7$. It is seen that the best improvements occur from $\theta_i$ to $\theta_{i+1}$ for $i=2, 3$. Thus, the best case scenario of uncertainty quantification obtained for this analysis is that of $\theta_4$.\\
Finally, we present results in Table  \ref{tab:my_label3} obtained from solving the inverse problem for $\theta_4$ using the data given in \cite{TDH}. The remaining parameters are fixed using the nominal parameters above and $T^*$ is the number of days used in the simulation
\begin{table}[th]
    \centering
    \begin{tabular}{|c|c|c|c|}
    \hline
      Parameters   &  Estimates & Standard Errors & Coefficients of Variation\\
      \hline
       $\beta_S$  & 9.9792e-01 & 1.8476e-01 & 1.8514e-01 \\
       \hline
       $\gamma_1$ & 7.2723e-02 & 6.2483e-03 & 8.5918e-02\\
       \hline
       $\gamma_2$ & 2.9536e-01 & 2.2578e-02 &  7.6443e-02\\
       \hline
       $\rho$ & 1.9201e-02 & 4.9837e-04 & 2.5956e-02\\
       \hline
       $\mu_4$ & 2.5413e-04 & 2.7412e-04 & 1.0786e+00 \\
       \hline
       $\eta$ & 6.7561e-01 & 4.2262e-02 & 6.2554e-02\\
       \hline
       $\alpha$ & 9.8887e-01 & 2.9311e-02 & 2.9641e-02\\
       \hline
    \end{tabular}
    \caption{Final parameter estimates from the data given in \cite{TDH} with $T^* = 100$}
    \label{tab:my_label3}
\end{table}
\section{Results and Discussions}
Seven inverse problems for California, Florida, Georgia, Maryland, Texas, Washington and  Wisconsin were solved to estimate the full parameter sets of the model. As seen in fig.\,5, the fits are reasonably good even for states like Tennessee and Wisconsin whose current infected population begins to flatten. Table \ref{tab:my_label6} (see Appendix) shows the fit parameter sets for each of the states with $T^*$ being the number of days used in the simulation. We see that the contact rates $\beta_A^\alpha$ and $\beta_S^\alpha$ for the asymptomatic and symptomatic population lies within 0.5--1.5 days$^{-1}$, a range suggested by Li \etal \cite{Li2020c}, Read \etal \cite{Read2020}, Shen \etal \cite{Shen2020}, Eikenberry \etal \cite{Eikenberry2020}. The incubation period $1/\sigma^\alpha$ is seen to be relatively small, around 1--2 days. The average length of active infection $1/\gamma_1^\alpha$ and $1/\gamma_2^\alpha$ of the asymptomatic and symptomatic population, respectively, is seen to be around 1--27 days which is in line with suggested range of days  in the literature \cite{Tang2020, Zhou2020}. The recovery rate via hospitalization is seen to be very small for California, Florida, Georgia and Washington where the data records little or no recovery at all for infected patients. The death rates $\mu_1^\alpha, \mu_2^\alpha, \mu_3^\alpha, \mu_4^\alpha$ from the exposed, asymptomatic infected, symptomatic infected and hospitalized compartments are seen to be relatively small ranging from 0--0.01 day$^{-1}$. With a $\mu_4^\alpha$ value of around 0, Florida and Georgia has no death from the hospitalized compartments. This is reasonable since $\rho^\alpha$ is also practically 0 indicating no recovery for this states as seen in the data (not shown here). \\
Table \ref{tab:my_label4} shows the basic reproduction number $R_0$ values computed for our model using (\ref{eqn:Ro}) and the parameters  listed in Table \ref{tab:my_label6}. The epidemic is expected to continue indefinitely if $R_0 > 1$ as predicted for all states. This suggest that stricter measures such as the use of masks in public places, social distancing, contact tracing and even longer stay-at-home orders need to be enforced in order to eradicate the epidemic. Fig.\,\ref{fig:sim1} shows plot of the infected (asymptomatic plus symptomatic), recovered and deaths.  We see that shape of most curves are similar except for Wisconsin and Tennessee where the infected population begin to flatten over time. In particular, spikes of the curves in fig.\,\ref{fig:sim1} begin around late March or early April 2020 and continues to rise.\\
{Table \ref{tab:my_label4b} shows a computational comparison of the proposed model and the integer-order counterpart using the sum squared errors (SSE) obtained for the infected, recovered and dead populations. As seen in the table, the fractional-order model shows superiority over the integer-order model in fitting the COVID-19 pandemic.} In fig.\,\ref{fig:sim3}, we show the cumulative infected and hospitalized population. This plot also shows that while the infected population flatten over time for states like California, Florida, Maryland, Tennessee, Texas and Wisconsin, it increases for Georgia and Washington. If drastic measures are not taken, the model suggest that many will be infected, at some point, before the end of 2021 in these states.
\begin{table}
    \centering
    \begin{tabular}{|c|c|}
     \hline
       States  & $R_0$ \\
         \hline
        California & 1.0837 \\
         \hline
        Florida & 1.1340 \\
         \hline
        Georgia & 1.1096 \\
         \hline
        Maryland & 1.2501 \\
         \hline
        Texas& 1.2367 \\
         \hline
        Tennessee & 1.0343 \\
         \hline
        Washington & 1.0964 \\
         \hline
        Wisconsin & 1.0754\\
         \hline
    \end{tabular}
    \caption{The basic reproduction number $R_0$ values for the selected states.\label{tab:my_label4} }
\end{table}

\begin{table}
{
\begin{tabular}{|c|c|c|}
\hline
States  & Integer-order model & Fractional-order model\\
\hline
 California & 8.323e+09 &  2.898e+09\\
 \hline
 Florida & 3.688e+09 & 2.544e+09\\
 \hline
 Georgia & 2.995e+09 & 7.464e+08\\
 \hline
 Maryland & 6.277e+08 & 6.551e+08\\
 \hline
 Tennessee & 3.308e+08 & 1.054e+08\\
 \hline
 Texas & 5.473e+09 & 2.712e+09\\
 \hline
 Washington & 8.780e+08 & 3.103e+08\\
 \hline
 Wisconsin & 7.404e+08 & 3.165e+08\\ 
 \hline
\end{tabular}
}
\caption{ {Comparison of sum squared errors (SSE) of the proposed model to the integer-order counterpart\label{tab:my_label4b}}
}
\end{table}

\begin{figure}
    \centering
    \includegraphics[width=1.2\textwidth]{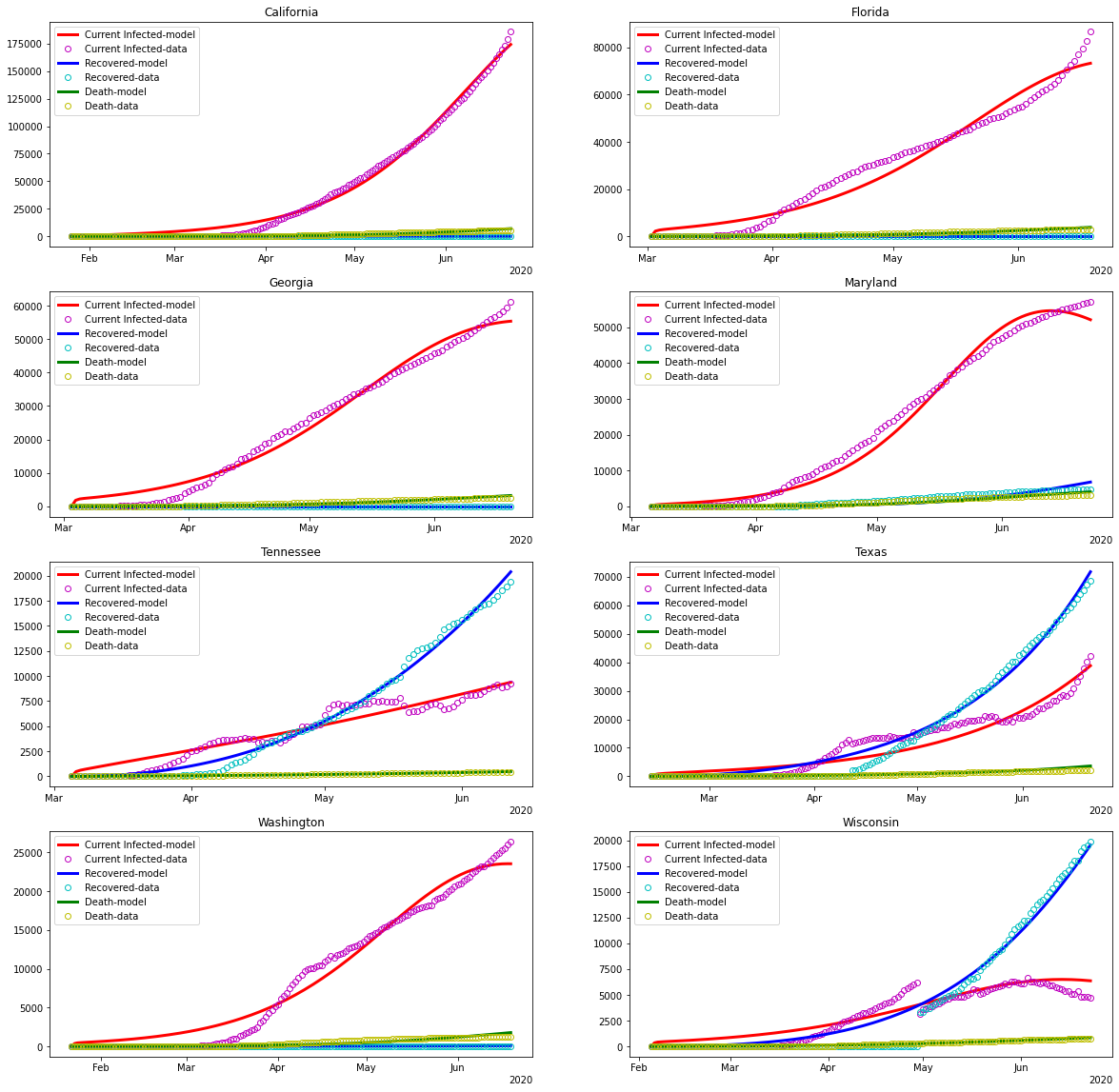}
    \caption{Model fits of the compartmental model to infected, recovered and deaths.}
    \label{fig:sim1}
\end{figure}
\begin{figure}[!htbp]
    \centering
    \includegraphics[width=1.2\textwidth]{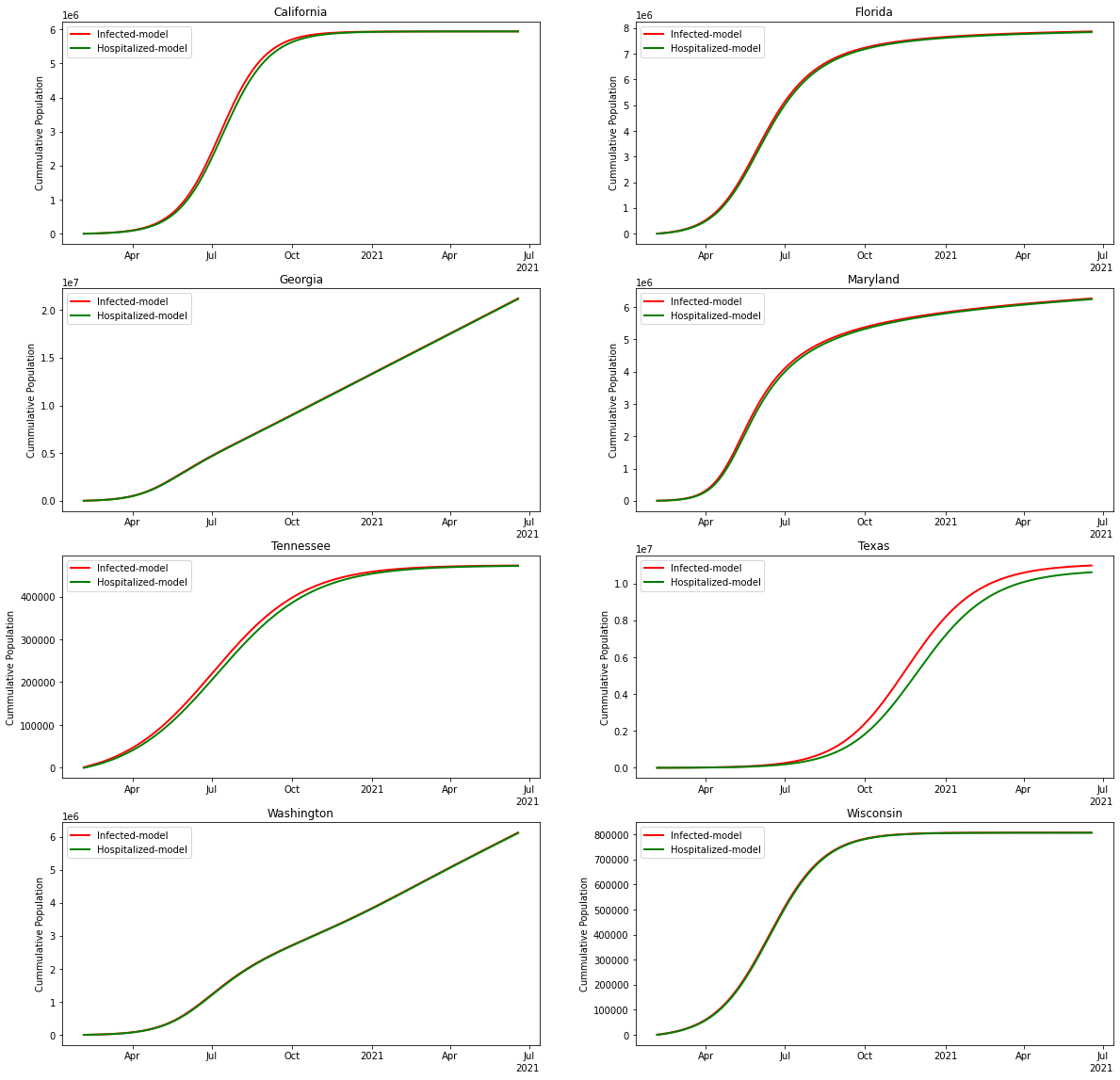}
    \caption{Model prediction for cumulative infected and hospitalized populations}
    \label{fig:sim3}
\end{figure}
\section{Conclusions}
A fractional-order compartmental  model is proposed to study the spread and dynamics of the {COVID-19} pandemic. We studied the properties of the model and discussed the parameters of the model in the  context of sensitivity and identifiability. We solve several inverse problems to estimate the fit parameters of the model using data from John Hopkins University \cite{JHU} for some selected states in the US. The basic reproduction number $R_0$ of the model was computed and shows that the states considered have $R_0$ values slightly greater than the critical value of one. The model suggests that stricter or aggressive measures need be enforced in order to slow the spread of the virus.
{
\section*{Acknowledgement}
The authors are grateful to the anonymous referees for their constructive criticisms and insightful suggestions in making the manuscript a better one. 
}
 \bibliography{analytical}
 \newpage
 \begin{center}
     \textbf{Appendices}
 \end{center}
\section*{Appendix A. Proof of Theorem 3.1.} 
\begin{proof}
By applying \cite[Theorem 3.1]{Lin2007}, we obtain the existence of the solutions. To show the uniqueness and boundedness of solutions, it suffices to show, by \cite[Remark 3.2]{Lin2007} and Rademacher's theorem, that $F = (f_1, f_2, f_3, f_4, f_5, f_6)$ is locally Lipschitz continuous where
\begin{equation*}
    \begin{split}
        f_1 &=  b_1^\alpha N - \dfrac{S}{N}\left(\beta_A^\alpha I_A + \beta_S^\alpha I_{S}\right) - b_1^\alpha S,\\
        f_2 &= \dfrac{S}{N}\left(\beta_A^\alpha I_A + \beta_S^\alpha I_{S}\right) - (\sigma^\alpha + \mu_1^\alpha) E,\\
        f_3 &= \eta\sigma^\alpha E - (\gamma_1^\alpha + \mu_2^\alpha)I_A, \\
        f_4 &= (1-\eta)\sigma^\alpha E - (\gamma_2^\alpha + \mu_3^\alpha) I_S,\\
        f_5 &= \gamma_1^\alpha I_A + \gamma_2^\alpha I_S  - (\rho^\alpha + \mu_4^\alpha) H,\\
        f_6 &= \rho^\alpha H - b_1^\alpha R.
    \end{split}
\end{equation*}
Let $X = (S, E, I_A, I_S, H, R)$, $\tilde{X} = (\tilde{S}, \tilde{E}, \tilde{I}_A, \tilde{I}_S, \tilde{H}, \tilde{R})$ and $||\cdot||$ denote the $L_2$ norm, then
\begin{equation*}
    \begin{split}
        ||F(X) - F(\tilde{X})|| &\leq ||f_1(X) - f_1(\tilde{X})|| +  ||f_2(X) - f_2(\tilde{X})||\\
        &~+ ||f_3(X) - f_3(\tilde{X})|| +
||f_4(X) - f_4(\tilde{X})||\\
&~+
||f_5(X) - f_5(\tilde{X})|| + 
||f_6(X) - f_6(\tilde{X})||\\
&\leq L ||X - \tilde{X}||,
    \end{split}
\end{equation*} 
where L = $\smash{\displaystyle\max_{1 \leq i \leq 6}}L_i$ and $L_1 = b_1^\alpha (N +1)+ \beta^\alpha_A + \beta^\alpha_S$, $L_2 = \beta^\alpha_A + \beta^\alpha_S + \sigma^\alpha + \mu_1^\alpha$, $L_3 = \eta\sigma^\alpha + \gamma_1^\alpha + \mu_2^\alpha$, $L_4 = (1-\eta)\sigma^\alpha + \gamma_2^\alpha + \mu^\alpha_3$, $L_5 = \gamma_1^\alpha + \gamma_2^\alpha + \rho^\alpha + \mu_4^\alpha$ and $L_6 = \rho^\alpha + b_1^\alpha$. Thus, $F$ satisfies the local Lipschitz conditions with respect to $X$ which proves the uniqueness and boundedness of solution to (\ref{eqn:model2}). Next we show the non-negativity of solutions. At first, we consider moving along the $S$-axis, that is $E(0) = I_A(0) = I_S(0) = H(0) = R(0) = 0$ and $0 < S(0) = S_0 \leq N$, then
\[ _0\mathcal{D}^\alpha_{t}S(t) = b_1^\alpha N - b_2^\alpha S\]
whose solution is given as 
\[ S(t) = S_0 E_{\alpha, 1}(-b_1^\alpha t^\alpha) + b_1^\alpha N t^\alpha E_{\alpha, \alpha+1}(-b_1^\alpha t^\alpha) > 0\]
since $b_1^\alpha > 0$ and $t>0$. In a similar manner, moving along each of the other respective axis (that is all initial conditions are zeros except for the axis being considered), it is easy to show that 
\begin{equation*}
    \begin{split}
        E(t) &= E_{\alpha, 1}\left( -(\sigma^\alpha + \mu_1^\alpha)t^\alpha\right)E_0 \geq 0\\
        I_A(t) &= E_{\alpha, 1}\left( -(\gamma_1^\alpha + \mu_2^\alpha)t^\alpha\right)I_{A0} \geq 0\\
        I_S(t) &= E_{\alpha, 1}\left( -(\gamma_2^\alpha + \mu_3^\alpha)t^\alpha\right)I_{S0} \geq 0\\
        H(t) &= E_{\alpha, 1}\left( -(\rho^\alpha + \mu_4^\alpha)t^\alpha\right)H_{0} \geq 0\\
        R(t) &= E_{\alpha, 1}\left( - b_1^\alpha t^\alpha\right)R_0 \geq 0.\\
    \end{split}
\end{equation*}
Therefore, all axis are non-negative invariant. Now, if the solution of the system is positive in the $E-I_A-I_S-H-R$ plane, then let $S(t^*) = 0$, $E(t^*) >0$, $I_A(t^*) >0$, $I_S(t^*) >0$, $H(t^*) >0$ and $R(t^*) >0$ for some $t^*$ such that $S(t) < S(t^*)$. But
\[ _0\mathcal{D}^\alpha_{t}S\left|_{t = t^*}\right. = b_1^\alpha N > 0\]
in this plane. Using the mean value theorem for Caputo-fractional derivative
\[ S(t) - S(t^*) = \dfrac{1}{\Gamma(\alpha)}\mathcal{D}^\alpha_{t}(\tau) (t - t^*)^\alpha\]
for some $\tau \in [t^*, t)$, we see that  $S(t) > S(t^*)$. This contradicts our previous statement. Similar arguments can be used for each of the remaining population variables.
\end{proof}
\section*{Appendix B. Proof of Theorem 3.2}
\begin{proof}
The Jacobian matrix evaluated at the DFE point is given as 
\begin{equation*}
    J = \left[
    \begin{array}{cccccc}
    -b_1^\alpha & 0 & -\beta_A^\alpha & -\beta_S^\alpha & 0 & 0\\
    0 & -(\sigma^\alpha + \mu_1^\alpha) & \beta_A^\alpha & \beta_S^\alpha & 0 &0\\
    0 & \eta \sigma^\alpha & -(\gamma_1^\alpha + \mu_2^\alpha) & 0 & 0 & 0 \\
    0 & (1-\eta)\sigma^\alpha &0 & -(\gamma_2^\alpha + \mu_3^\alpha) & 0 & 0\\
    0 & 0 & \gamma_1^\alpha & \gamma_2^\alpha & -(\rho^\alpha + \mu_4^\alpha) & 0\\
    0 & 0 & 0 & 0 & \rho^\alpha & -b_1^\alpha
    \end{array}
    \right].
\end{equation*}
The eigenvalues of the Jacobian matrix are given as $\lambda_1 = \lambda_2 = -b_1^\alpha$, $\lambda_3 = -(\rho^\alpha + \mu^\alpha_4)$ and the roots of the equation $z^3 + Az^2 + Bz + C$, where \\
$A = (\sigma^\alpha + \gamma_1^\alpha + \gamma_2^\alpha + \mu_1^\alpha + \mu_2^\alpha + \mu_3^\alpha)$,\\
$B = - \sigma^\alpha (\eta\beta_A^\alpha + (1-\eta)\beta_S^\alpha) + (\sigma ^\alpha + \mu_1^\alpha)(\gamma_1^\alpha + \mu_2^\alpha) + (\sigma ^\alpha + \mu_1^\alpha)(\gamma_2^\alpha + \mu_3^\alpha) + (\gamma_1^\alpha + \mu_2^\alpha)(\gamma_2^\alpha + \mu_3^\alpha)$,\\
$C = - \sigma^\alpha(\eta(\gamma_2^\alpha + \mu_3^\alpha)\beta_A^\alpha + (1-\eta)(\gamma_1^\alpha + \mu_2^\alpha)\beta_S^\alpha) + (\sigma^\alpha + \mu_1^\alpha)(\gamma_1^\alpha + \mu_2^\alpha)(\gamma_2^\alpha + \mu_3^\alpha)$. To show stability, we apply the Routh-Hurwitz criterion. Clearly $A> 0$ and  $C>0$ if $R_0 < 1$. It is easy to show that  $AB > C$ if $K < 1$ which completes the proof.
\end{proof}
\section*{Appendix C. Proof of Theorem 3.3}
\begin{proof}
The Jacobian matrix evaluated at the EE point is given as 
\begin{equation*}
    J = \left[
    \begin{array}{cccccc}
    -b_1^\alpha - b_1^\alpha\left(1  - \dfrac{1}{R_0}\right)& 0 & -\dfrac{ \beta_A^\alpha}{R_0} & -\dfrac{ \beta_S^\alpha}{ R_0} & 0 & 0\\
    b_1^\alpha\left(1  - \dfrac{1}{R_0}\right) & -(\sigma^\alpha + \mu_1^\alpha) & \dfrac{\beta_A^\alpha}{R_0} & \dfrac{\beta_S^\alpha}{R_0} & 0 &0\\
    0 & \eta \sigma^\alpha & -(\gamma_1^\alpha + \mu_2^\alpha) & 0 & 0 & 0 \\
    0 & (1-\eta)\sigma^\alpha &0 & -(\gamma_2^\alpha + \mu_3^\alpha) & 0 & 0\\
    0 & 0 & \gamma_1^\alpha & \gamma_2^\alpha & -(\rho^\alpha + \mu_4^\alpha) & 0\\
    0 & 0 & 0 & 0 & \rho^\alpha & -b_1^\alpha
    \end{array}
    \right].
\end{equation*}
The eigenvalues of the Jacobian matrix are $\lambda_1 = -b_1^\alpha, ~\lambda_2 = - (\rho^\alpha + \mu_4^\alpha)$ and the roots of the equation $z^4 + A_1z^3 + A_2z^2 + A_3z +  A_4$. By Applying the Routh-Hurwitz criterion, the EE is stable if  $A_1 > 0$, $A_4 > 0$ and $(A_1A_2 -A_3)A_3 -A_1^2A_4 \geq 0$. Clearly $A_1 > 0$ since $R_0 > 1$. Also, $A_4 > 0$ implies that
$$ b_1^{3\alpha}R_0^3(2R_0 -1)(\sigma^\alpha + \mu_1^\alpha)(\gamma_1^\alpha + \mu_2^\alpha)(\gamma_2^\alpha + \mu_3^\alpha) > \sigma^\alpha((1-\eta)\beta_S^\alpha(\gamma_1^\alpha + \mu_2^\alpha) + \eta\beta_A^\alpha(\gamma_2^\alpha + \mu_3^\alpha))]lpha(\gamma_2^\alpha + \mu_3^\alpha))$$
$\Rightarrow ~~~b_1^{3\alpha} R_0^3(2R_0 - 1)> \dfrac{\sigma^\alpha b_1^\alpha((1-\eta)\beta_S^\alpha(\gamma_1^\alpha + \mu_2^\alpha) + \eta\beta_A^\alpha(\gamma_2^\alpha + \mu_3^\alpha))}{(\sigma^\alpha + \mu_1^\alpha)(\gamma_1^\alpha + \mu_2^\alpha)(\gamma_2^\alpha + \mu_3^\alpha)}$\\
$\Rightarrow ~~~b_1^{3\alpha} R_0^2(2R_0 - 1) > 1$ which is true since $R_0 > 1$
\end{proof}
\section*{Appendix D. Data Sets and Methods}
Data for cumulative confirmed cases, recovered and deaths were obtained from the Tennessee Department of Health (TDH) and Johns Hopkins University (JHU) Center for Systems Science and Engineering which are made available to the public on TDH's website \cite{TDH} and Github \cite{JHU}, respectively. The raw files are converted into Panda data frames and stored for ease of access.  For the files obtained from \cite{JHU}, the data are sorted according to counties, so we aggregate the cases for each of the recorded counties to obtain the total number of cases in each day for the states. For the total number of population for each states, we used the data obtained from the US Census Bureau \cite{USCB}.\\
All numerical simulations were done in python using our numerical scheme \cite{Biala2018} from  which we obtain the solution of the proposed model at each time step as
\begin{enumerate}
 \item Predictor:\\ 
    \begin{equation*}
    \begin{split}
        S_p &= S_{j} + \dfrac{\tau^\alpha}{\Gamma(1 + \alpha)}F_1(t_j, S_j, E_j, I_{A, j}, I_{S, j}, {H}_j, R_j, D_j) + \tilde{H}_{1, j}\\
         E_p &= E_{j} + \dfrac{\tau^\alpha}{\Gamma(1 + \alpha)}F_2(t_j, S_j, E_j, I_{A, j}, I_{S, j}, {H}_j, R_j, D_j) + \tilde{H}_{2,j}\\
          I_{A,p} &= I_{A, j} + \dfrac{\tau^\alpha}{\Gamma(1 + \alpha)}F_3(t_j, S_j, E_j, I_{A, j}, I_{S, j}, H_j, R_j, D_j) + \tilde{H}_{3,j}\\
          I_{S,p} &= I_{S, j} + \dfrac{\tau^\alpha}{\Gamma(1 + \alpha)}F_4(t_j, S_j, E_j, I_{A, j}, I_{S, j}, H_j, R_j, D_j) + \tilde{H}_{4,j}\\
           H_p &= H_{j} + \dfrac{\tau^\alpha}{\Gamma(1 + \alpha)}F_5(t_j, S_j, E_j, I_{A, j}, I_{S, j}, H_j, R_j, D_j) + \tilde{H}_{5,j}\\
            R_p &= R_{j} + \dfrac{\tau^\alpha}{\Gamma(1 + \alpha)}F_6(t_j, S_j, E_j, I_{A, j}, I_{S, j}, H_j, R_j, D_j) + \tilde{H}_{6,j}\\
             D_p &= D_{j} + \dfrac{\tau^\alpha}{\Gamma(1 + \alpha)}F_7(t_j, S_j, E_j, I_{A, j}, I_{S, j}, H_j, R_j, D_j) + \tilde{H}_{7,j}\\
    \end{split}
\end{equation*}
    \item Corrector:\\ 
    \begin{equation*}
    \begin{split}
        S_{j+1} &= S_{j} + \dfrac{\tau^\alpha}{\Gamma(2 + \alpha)}\Big(\alpha\, F_1(t_j, S_j, E_j, I_{A, j}, I_{S, j}, H_j, R_j, D_j) \\
        & ~~~~~~~+ ~F_1(t_{j+1}, S_p, E_p, I_{A, p}, I_{S, p}, H_p, R_p, D_p)\Big) + \tilde{H}_{1, j},\\
        E_{j+1} &= E_{j} + \dfrac{\tau^\alpha}{\Gamma(2 + \alpha)}\Big(\alpha\, F_2(t_j, S_j, E_j, I_{A, j}, I_{S, j}, H_j, R_j, D_j) \\
        & ~~~~~~~+ ~F_2(t_{j+1}, S_p, E_p, I_{A, p}, I_{S, p}, H_p, R_p, D_p)\Big) + \tilde{H}_{2, j},\\
        I_{A, j+1} &= I_{A, j} + \dfrac{\tau^\alpha}{\Gamma(2 + \alpha)}\Big(\alpha\, F_3(t_j, S_j, E_j, I_{A, j}, I_{S, j}, H_j, R_j, D_j) \\
        & ~~~~~~~+ ~F_3(t_{j+1}, S_p, E_p, I_{A, p}, I_{S, p}, H_p, R_p, D_p)\Big) + \tilde{H}_{3, j},\\
        I_{S,j+1} &= I_{S,j} + \dfrac{\tau^\alpha}{\Gamma(2 + \alpha)}\Big(\alpha\, F_4(t_j, S_j, E_j, I_{A, j}, I_{S, j}, H_j, R_j, D_j) \\
        & ~~~~~~~+ ~F_4(t_{j+1}, S_p, E_p, I_{A, p}, I_{S, p}, H_p, R_p, D_p)\Big) + \tilde{H}_{4, j},\\
        H_{j+1} &= H_{j} + \dfrac{\tau^\alpha}{\Gamma(2 + \alpha)}\Big(\alpha\, F_5(t_j, S_j, E_j, I_{A, j}, I_{S, j}, H_j, R_j, D_j) \\
        & ~~~~~~~+ ~F_5(t_{j+1}, S_p, E_p, I_{A, p}, I_{S, p}, H_p, R_p, D_p)\Big) + \tilde{H}_{5, j},\\
        R_{j+1} &= R_{j} + \dfrac{\tau^\alpha}{\Gamma(2 + \alpha)}\Big(\alpha\, F_6(t_j, S_j, E_j, I_{A, j}, I_{S, j}, H_j, R_j, D_j) \\
        & ~~~~~~~+ ~F_6(t_{j+1}, S_p, E_p, I_{A, p}, I_{S, p}, H_p, R_p, D_p)\Big) + \tilde{H}_{6, j},\\
        D_{j+1} &= D_{j} + \dfrac{\tau^\alpha}{\Gamma(2 + \alpha)}\Big(\alpha\, F_7(t_j, S_j, E_j, I_{A, j}, I_{S, j}, H_j, R_j, D_j) \\
        & ~~~~~~~+ ~F_7(t_{j+1}, S_p, E_p, I_{A, p}, I_{S, p}, H_p, R_p, D_p)\Big) + \tilde{H}_{7, j},\\
    \end{split}
\end{equation*}
where 
\begin{equation*}
    \begin{split}
        F_1(t_j, S_j, E_j, I_{A, j}, I_{S, j}, H_j, R_j, D_j) &=  b_1^\alpha N - \dfrac{S_j}{N}\left(\beta_A^\alpha I_{A,j} + \beta_S^\alpha I_{S, j}\right) - b_1^\alpha S_j,\\
        F_2(t_j, S_j, E_j, I_{A, j}, I_{S, j}, H_j, R_j, D_j) &= \dfrac{S_j}{N}\left(\beta_A^\alpha I_{A,j} + \beta_S^\alpha I_{S, j}\right) - (\sigma^\alpha + \mu_1^\alpha) E_j,\\
       F_3(t_j, S_j, E_j, I_{A, j}, I_{S, j}, H_j, R_j, D_j) &= \eta\sigma^\alpha E_j - (\gamma_1^\alpha + \mu_2^\alpha)I_{A,j}, \\
        F_4(t_j, S_j, E_j, I_{A, j}, I_{S, j}, H_j, R_j, D_j) &= (1-\eta)\sigma^\alpha E_j - (\gamma_2^\alpha + \mu_3^\alpha) I_{S,j},\\
        F_5(t_j, S_j, E_j, I_{A, j}, I_{S, j}, H_j, R_j, D_j) &= \gamma_1^\alpha I_{A,j} + \gamma_2^\alpha I_{S,j}  - (\rho^\alpha + \mu_4^\alpha) H_j,\\
        F_6(t_j, S_j, E_j, I_{A, j}, I_{S, j}, H_j, R_j, D_j) &= \rho^\alpha H_j - b_1^\alpha R_j,\\
        F_7(t_j, S_j, E_j, I_{A, j}, I_{S, j}, H_j, R_j, D_j) &= \mu_1^\alpha E_j + \mu_2^\alpha I_{A,j} + \mu_3^\alpha I_{S,j} + \mu_4^\alpha H_{j},
    \end{split}
\end{equation*}
and 
\begin{equation*}
    \begin{split}
        \tilde{H}_{1, j} = \dfrac{\tau^\alpha}{\Gamma(2 + \alpha)}\sum_{l=0}^{j}a_{_{l, j}}\,F_1(t_l, S_l, E_l, I_{A, l}, I_{S, l}, H_l, R_l, D_l),\\
        \tilde{H}_{2, j} = \dfrac{\tau^\alpha}{\Gamma(2 + \alpha)}\sum_{l=0}^{j}a_{_{l, j}}\,F_2(t_l, S_l, E_l, I_{A, l}, I_{S, l}, H_l, R_l, D_l),\\
        \tilde{H}_{3, j} = \dfrac{\tau^\alpha}{\Gamma(2 + \alpha)}\sum_{l=0}^{j}a_{_{l, j}}\,F_3(t_l, S_l, E_l, I_{A, l}, I_{S, l}, H_l, R_l, D_l),\\
        \tilde{H}_{4, j} = \dfrac{\tau^\alpha}{\Gamma(2 + \alpha)}\sum_{l=0}^{j}a_{_{l, j}}\,F_4(t_l, S_l, E_l, I_{A, l}, I_{S, l}, H_l, R_l, D_l),\\
        \tilde{H}_{5, j} = \dfrac{\tau^\alpha}{\Gamma(2 + \alpha)}\sum_{l=0}^{j}a_{_{l, j}}\,F_5(t_l, S_l, E_l, I_{A, l}, I_{S, l}, H_l, R_l, D_l),\\
        \tilde{H}_{6, j} = \dfrac{\tau^\alpha}{\Gamma(2 + \alpha)}\sum_{l=0}^{j}a_{_{l, j}}\,F_6(t_l, S_l, E_l, I_{A, l}, I_{S, l}, H_l, R_l, D_l),\\
        \tilde{H}_{7, j} = \dfrac{\tau^\alpha}{\Gamma(2 + \alpha)}\sum_{l=0}^{j}a_{_{l, j}}\,F_7(t_l, S_l, E_l, I_{A, l}, I_{S, l}, H_l, R_l, D_l)
    \end{split}
\end{equation*}
are the memory terms of the respective population variables and 
\[ a_{_{l,j}} =  \dfrac{\tau^\gamma }{\Gamma(\gamma  + 2)}
\begin{cases}
-(j- \gamma )(j+1)^\gamma  + j^\gamma (2j - \gamma  -1 ) - (j-1)^{\gamma  +1 }, \qquad \qquad \qquad  \qquad l = 0,\\
(j-l+2)^{\gamma +1} - 3(j-l+1)^{\gamma +1} + 3(j-l)^{\gamma  +1 }- (j-l-1)^{\gamma  +1 }, \quad 1 \leq l \leq j-1,\\
2^{\gamma  +1 } - \gamma   - 3, \qquad \qquad \qquad  \qquad  \qquad \qquad \qquad \qquad  \qquad  \qquad \qquad  \qquad l = j.
\end{cases}
\]
The fitted parameters were obtained by using python's \texttt{scipy.optimize.minimize} routine with the limited memory BFGS method. One main benefit of the routine is the use of bounds for fit parameters. This allows faster convergence of the algorithm and ensures obtaining meaningful fit parameters of the model. The parameters are constrained to lie between 0 and 1 except for the transmission rates that are allowed to lie between 0.5 and 1.5 as suggested in \cite{Li2020c, Read2020, Shen2020, Eikenberry2020}. In some cases like Florida and Georgia where the data shows little or no recovery of the infected population, we constrain the fit parameter to lie within 0 and 1e-10. Fits were made comparing the simulation results  and the data obtained from \cite{JHU} for the infected, recovered and deaths.
\end{enumerate}
\section*{Appendix E. Tables}
\begin{landscape}
\begin{table}[H]
    \centering
    \begin{tabular}{c|c|c|c|c|c|c|c|c|c|c|c|c|c}
         & $b_1$ & $\beta_A$ & $\beta_S$ & $\sigma$ & $\gamma_1$ & $\gamma_2$ & $\rho$ & $\mu_1$ & $\mu_2$ & $\mu_3$  & $\mu_4$ & $\eta$ & $\alpha$\\
        \hline
        $\tilde{\theta}$ &1.13e-4 & 2.60e-4 & 1.23e+0 & 9.96e-1 & 2.69e-2 & 1.00e+0 & 6.18e-3 & 3.32e-3 & 1.59e-5 & 3.05e-4 & 2.13e-5 & 1.57e-1 & 1.00e+0\\
        $E$ & 3.16e-6 & 1.42e-4 & 5.74e-3 & 2.29e-3 & 7.39e-4 & 4.56e-3 & 1.01e-5 & 4.94e-3 & 7.37e-4 & 6.48e-3 & 5.52e-5 & 9.47e-4 & 9.35e-4\\
        $V$ & 2.81e-2 & 5.49e-1 & 4.66e-3 & 2.30e-3 & 2.75e-2 & 4.56e-3 & 1.64e-3 & 1.49e+0 & 4.64e+1 & 2.12e+1 & 2.59e+0 & 6.04e-3 & 9.35e-4\\
        \hline
        $\tilde{\theta}$ &3.52e-4 & & 1.23e+0 & 9.96e-1 & 2.74e-2 & 1.00e+0 & 6.17e-3 & 3.31e-3 &  & 2.26e-4 & 2.42e-5 & 1.57e-1 & 1.00e+0\\
        $E$ & 3.09e-6 &  & 1.46e-3 & 1.82e-3 & 1.37e-4 & 1.19e-3 & 4.11e-6 & 1.43e-3 & & 1.75e-3 & 7.48e-6 & 4.59e-4 & 3.53e-4\\
        $V$ & 8.78e-3 &  & 1.19e-3 & 1.83e-3 & 4.99e-3  & 1.19e-3 & 6.67e-4 & 4.31e-1 &  & 7.75e+0 & 3.09e-1 & 2.94e-3 & 3.53e-4\\
        \hline
         $\tilde{\theta}$ && & 1.23e+0 & 9.96e-1 & 2.73e-2 & 9.97e-1 & 6.16e-3 & 3.27e-3 &  &  & 2.15e-5 & 1.59e-1 & 9.96e-1\\
        $E$ & &  & 1.41e-3 & 1.72e-3 & 1.30e-4 & 1.11e-3 & 3.94e-6 & 1.76e-4 & &  & 6.92e-6 & 4.13e-4 & 3.32e-4\\
        $V$ & &  & 1.15e-3 & 1.73e-3 & 4.73e-3  & 1.12e-3 & 6.39e-4 & 5.37e-2 &  &  & 3.21e-1 & 2.60e-3 & 3.33e-4\\
        \hline
         $\tilde{\theta}$ && & 1.23e+0 & & 2.73e-2 & 9.97e-1 & 6.16e-3 &  &  &  & 1.85e-5 & 1.59e-1 & 9.96e-1\\
        $E$ & &  & 1.39e-3 &  & 1.13e-4 & 1.06e-3 & 1.85e-6 & & & & 1.23e-6 & 3.71e-4 & 1.04e-4\\
        $V$ & &  & 1.13e-3 & & 4.14e-3  & 1.06e-3 & 3.00e-4 &  &  &  & 6.64e-2 & 2.33e-3 & 1.04e-4\\
        \hline
         $\tilde{\theta}$ && &  & & 2.68e-2 & 9.98e-1 & 6.15e-3 &  &  &  & 1.98e-05 & & 9.97e-1\\
        $E$ & &  & &  & 1.94e-5 & 2.28e-5 & 1.35e-6 & & & &  1.24e-06 &  & 4.57e-5\\
        $V$ & &  &  & & 7.25e-4  & 2.28e-5 & 2.20e-4 &  &  &  & 6.24e-02 &  & 4.59e-5\\
        \hline
         $\tilde{\theta}$ && &  & & 2.68e-2 & 9.99e-1 & 6.15e-3 &  &  &  & & & 9.98e-1\\
        $E$ & &  & &  & 1.87e-5 & 2.27e-5 & 1.33e-6 & & & & &  & 4.55e-5\\
        $V$ & &  &  & & 7.01e-4  & 2.27e-5 & 2.16e-4 &  &  &  & &  & 4.57e-5\\
        \hline
         $\tilde{\theta}$ && &  & &  & 9.98e-1 & 6.16e-3 &  &  &  & & & 9.97e-1\\
        $E$ & &  & &  & & 2.26e-5 & 1.33e-6 && &  & &  & 4.51e-5\\
        $V$ & &  &  & &   & 2.26e-5 & 2.15e-4 & & &  &  &  & 4.53e-4\\
        \hline
    \end{tabular}
    \caption{Parameter estimates for solving seven inverse problems from a synthetic data generated using the given nominal parameters and variance. For each parameter subset, we display the estimate $(\tilde{\theta})$, the standard error, E and the coefficient of variation, $V = E/\tilde{\theta}$.    \label{tab:my_label5}}

\end{table}
\end{landscape}
\begin{landscape}
\begin{table}
    \centering
    \begin{tabular}{|c|c|c|c|c|c|c|c|}
    \hline
       $\theta$   & California  & Florida & Georgia & Maryland & Texas & Washington & Wisconsin \\
         \hline
        $b_1$  & 1.0000e-15 & 1.0000e-15 & 4.6495e-02 & 1.0080e-05 & 1.4196e-05 & 2.0116e-02 & 1.0000e-15\\
         \hline
        $\beta_A$  & 8.8973e-05 & 9.8043e-01 & 1.0939e+00  & 1.3360e+00 & 5.6047e-01 & 9.5295e-01 & 6.3325e-01\\
         \hline
        $\beta_S$  & 1.1213e+00  & 6.9601e-01 & 5.1042e-01 & 1.3320e+00 & 1.1989e-02 & 1.0000e+00 & 1.0773e+00\\
         \hline
        $\sigma$ & 9.9999e-01 & 1.0000e+00 & 7.5663e-01 & 1.0000e+00 & 1.0000e+00 & 7.3519e-01 & 7.1669e-01\\
         \hline
        $\gamma_1$  & 3.5769e-02 & 8.3577e-01 &  9.8459e-01 & 9.9686e-01 & 1.2284e-01 & 8.9911e-01 & 1.0000e+00\\
         \hline
        $\gamma_2$  & 9.1747e-01 & 7.8446e-01 & 7.6914e-01 & 9.6675e-01 & 6.6645e-02 & 8.9098e-01 & 1.0000e+00\\
         \hline
        $\rho$  & 1.3366e-07  & 1.0000e-15 & 1.0000e-15 & 4.5126e-05 & 5.4691e-02 & 2.9082e-06& 1.0325e-03\\
         \hline
        $\mu_1$ & 2.7409e-03 & 6.7600e-09 & 6.5427e-04 & 6.6585e-04 & 1.0000e-15 & 6.9886e-04 & 8.7380e-04 \\
         \hline
        $\mu_2$ & 1.1988e-05 & 8.1025e-04 & 3.4154e-04 & 1.0000e-15 & 1.2534e-02 & 3.1211e-06& 8.6293e-02\\
         \hline
        $\mu_3$  & 6.8398e-05 & 6.9384e-05 & 3.2020e-08 & 1.0000e-15 & 1.0000e-15 &3.6933e-06 & 6.0737e-04\\
         \hline
        $\mu_4$ & 1.0000e-15 & 1.0000e-15 & 1.0000e-15 & .0000e-15 & 1.0000e-15 & 1.2343e-05 & 1.0000e-15\\
         \hline
        $\eta$  & 1.1097e-01 &  9.4378e-01 & 1.0000e+00 & 9.9986e-01 & 3.0315e-01 & 3.9910e-01 & 4.6680e-05\\
         \hline
        $\alpha$  & 9.9980e-01 & 8.7834e-01 & 9.9967e-01 & 7.7423e-01 & 9.8525e-01  & 1.0000e+00 & 1.0000e+00 \\
         \hline
         $T^*$ & 150 & 110 & 110 & 110 & 130 & 150 & 140\\
         \hline
    \end{tabular}
    \caption{Model fit parameters for some selected states in the US, $T^*$ shows the number of days used for each state.\label{tab:my_label6}}
    
\end{table}
\end{landscape}
\end{document}